\documentclass[11pt]{article}
\usepackage{epsf}
\usepackage{amsmath}
\usepackage{epsfig}
\usepackage{times}
\usepackage{amssymb}
\usepackage{amsthm}
\usepackage{setspace}
\usepackage{cite}

\usepackage{algorithmic}  % This package provides an algorithmic environment fo describing algorithms.
\usepackage{algorithm}

\usepackage{shadow}
\usepackage{fancybox}
\usepackage{fancyhdr}

\def\w{{\bf w}}

\def\y{{\bf y}}

\def\x{{\bf x}}

\def\x{{\mathbf x}}

\def\w{{\bf w}}

\def\x{{\bf x}}
\def\y{{\bf y}}

\def\b{{\bf b}}

\def\h{{\bf h}}

\def\be{\begin{equation}}
\def\ee{\end{equation}}
\def\ba{\left[\begin{array}}
\def\ea{\end{array}\right]}

\def\w{{\bf w}}

\def\x{{\bf x}}
\def\y{{\bf y}}

\def\b{{\bf b}}

\def\1{{\bf 1}}

\def\g{{\bf g}}
\def\0{{\bf 0}}

\def\Sweak{S_{weak}}

\def\betaweak{\beta_{weak}}

\def\betasec{\beta_{sec}}

\def\Ssec{S_{sec}}

\def\Sstr{S_{str}}
\def\betastr{\beta_{str}}

\newtheorem{theorem}{Theorem}
\newtheorem{corollary}{Corollary}

\newtheorem{lemma}{Lemma}

\begin{document}

\begin{singlespace}

\title {Lifting $\ell_q$-optimization thresholds %A tight variant of Gordon's escape through a mesh theorem
%\footnote{ This work was supported in
%part.}
}
\author{
\textsc{Mihailo Stojnic}
\\
\\
{School of Industrial Engineering}\\
{Purdue University, West Lafayette, IN 47907} \\
{e-mail: {\tt mstojnic@purdue.edu}} }
\date{}
\maketitle

\centerline{{\bf Abstract}} \vspace*{0.1in}

In this paper we look at a connection between the $\ell_q,0\leq q\leq 1$, optimization and under-determined linear systems of equations with sparse solutions. The case $q=1$, or in other words $\ell_1$ optimization and its a connection with linear systems has been thoroughly studied in last several decades; in fact, especially so during the last decade after the seminal works \cite{CRT,DOnoho06CS} appeared. While current understanding of $\ell_1$ optimization-linear systems connection is fairly known, much less so is the case with a general $\ell_q,0<q<1$, optimization. In our recent work \cite{StojnicLqThrBnds10} we provided a study in this direction. As a result we were able to obtain a collection of lower bounds on various $\ell_q,0\leq q\leq 1$, optimization thresholds. In this paper, we provide a substantial conceptual improvement of the methodology presented in \cite{StojnicLqThrBnds10}. Moreover, the practical results in terms of achievable thresholds are also encouraging. As is usually the case with these and similar problems, the methodology we developed emphasizes their a combinatorial nature and attempts to somehow handle it. Although our results' main contributions should be on a conceptual level, they already give a very strong suggestion that $\ell_q$ optimization can in fact provide a better performance than $\ell_1$, a fact long believed to be true due to a tighter optimization relaxation it provides to the original $\ell_0$ sparsity finding oriented original problem formulation. As such, they in a way give a solid boost to further exploration of the design of the algorithms that would be able to handle $\ell_q,0<q<1$, optimization in a reasonable (if not polynomial) time.

\vspace*{0.25in} \noindent {\bf Index Terms: under-determined linear systems; sparse solutions; $\ell_q$-minimization}.

\end{singlespace}

%%%%%%%%%%%%%%%%%%%%%%%%%%%%%%%%%%%%%%%%%%%%%%%%%%%%%%%%%%%%%%%%%
\section{Introduction}
\label{sec:back}
%%%%%%%%%%%%%%%%%%%%%%%%%%%%%%%%%%%%%%%%%%%%%%%%%%%%%%%%%%%%%%%%%

Although the methods that we will propose have no strict limitations as to what structure they can handle we will restrict our attention to under-determined linear systems of equations with sparse solutions. As is well known in mathematical terms a linear system of equations can be written as
\begin{equation}
A\x=\y \label{eq:system}
\end{equation}
where $A$ is an $m\times n$ ($m<n$) system matrix and $\y$ is
an $m\times 1$ vector. Typically one is then given $A$ and $\y$ and the goal is to determine $\x$. However when ($m<n$) the odds are that there will be many solutions and that the system will be under-determined. In fact that is precisely the scenario that we will look at. However, we will slightly restrict our choice of $\y$. Namely, we will assume that $\y$ can be represented as
\begin{equation}
\y=A\tilde{\x}, \label{eq:yrepsystem}
\end{equation}
where we also assume that $\tilde{\x}$ is a $k$-sparse vector (here and in the rest of the paper, under $k$-sparse vector we assume a vector that has at most $k$ nonzero components). This essentially means that we are interested in solving (\ref{eq:system}) assuming that there is a solution that is $k$-sparse. Moreover, we will assume that there is no solution that is less than $k$-sparse, or in other words, a solution that has less than $k$ nonzero components.

These systems gained a lot of attention recently in first place due to seminal results of \cite{CRT,DOnoho06CS}. In fact particular types of these systems that happened to be of substantial mathematical interest are the so-called random systems. In such systems one models generality of $A$ through a statistical distribution. Such a concept will be also of our interest in this paper. To ease the exposition, whenever we assume a statistical context that will mean that the system (measurement) matrix $A$ has i.i.d. standard normal components. We also emphasize that only source of randomness will the components of $A$. Also, we do mention that all of our work is in no way restricted to a Gaussian type of randomness. However, we find it easier to present all the results under such an assumption. More importantly a great deal of results of many of works that will refer to in a statistical way also hold for various non-gaussian types of randomness. As for \cite{CRT,DOnoho06CS}, they looked at a particular technique called $\ell_1$ optimization and showed for the very first time that in a statistical context such a technique can recover a sparse solution (of sparsity linearly proportional to the system dimension). These results then created an avalanche of research and essentially could be considered as cornerstones of a field today called compressed sensing (while there is a tone of great work done in this area during the last decade, and obviously the literature on compressed sensing is growing on a daily basis, we instead of reviewing all of them refer to two introductory papers \cite{CRT,DOnoho06CS} for a further comprehensive understanding of their meaning on a grand scale of all the work done over the last decade).

Although our results will be easily applicable to any regime, to make writing in the rest of the paper easier, we will assume the \emph{typical}
so-called \emph{linear} regime, i.e. we will assume that $k=\beta n$
and that the number of equations is $m=\alpha n$ where
$\alpha$ and $\beta$ are constants independent of $n$ (more
on the non-linear regime, i.e. on the regime when $m$ is larger than
linearly proportional to $k$ can be found in e.g.
\cite{CoMu05,GiStTrVe06,GiStTrVe07}).

Now, given the above sparsity assumption, one can then rephrase the original problem (\ref{eq:system}) in the following way
\begin{eqnarray}
\mbox{min} & & \|\x\|_{0}\nonumber \\
\mbox{subject to} & & A\x=\y. \label{eq:l0}
\end{eqnarray}
Assuming that $\|\x\|_{0}$ counts how many nonzero components $\x$ has, (\ref{eq:l0}) is essentially looking for the sparsest $\x$ that satisfies (\ref{eq:system}), which, according to our assumptions, is exactly $\tilde{\x}$. Clearly, it would be nice if one can solve in a reasonable (say polynomial) time (\ref{eq:l0}). However, this does not appear to be easy. Instead one typically resorts to its relaxations that would be solvable in polynomial time. The first one that is typically employed is called $\ell_1$-minimization. It  essentially relaxes the $\ell_0$ norm in the above optimization problem to the first one that is known to be solvable in polynomial time, i.e. to $\ell_1$. The resulting optimization problem then becomes
\begin{eqnarray}
\mbox{min} & & \|\x\|_{1}\nonumber \\
\mbox{subject to} & & A\x=\y. \label{eq:l1}
\end{eqnarray}
Clearly, as mentioned above (\ref{eq:l1}) is an optimization problem solvable in polynomial time. In fact it is a very simple linear program. Of course the question is: how well does it approximate the original problem (\ref{eq:l0}). Well, for certain system dimensions it works very well and actually can find exactly the same solution as (\ref{eq:l0}). In fact, that is exactly what was shown in \cite{CRT,DOnoho06CS,DonohoPol}. A bit more specifically, it was shown in \cite{CRT} that if
$\alpha$ and $n$ are given, $A$ is given and satisfies the restricted isometry property (RIP) (more on this property the interested reader can find in e.g. \cite{Crip,CRT,Bar,Ver,ALPTJ09}), then
any unknown vector $\tilde{\x}$ in (\ref{eq:yrepsystem}) with no more than $k=\beta n$ (where $\beta$
is a constant dependent on $\alpha$ and explicitly
calculated in \cite{CRT}) non-zero elements can be recovered by
solving (\ref{eq:l1}). On the other hand
in \cite{DonohoUnsigned,DonohoPol} Donoho considered the polytope obtained by
projecting the regular $n$-dimensional cross-polytope $C_p^n$ by $A$. He then established that
the solution of (\ref{eq:l1}) will be the $k$-sparse solution of
(\ref{eq:system}) if and only if
$AC_p^n$ is centrally $k$-neighborly
(for the definitions of neighborliness, details of Donoho's approach, and related results the interested reader can consult now already classic references \cite{DonohoUnsigned,DonohoPol,DonohoSigned,DT}). In a nutshell, using the results
of \cite{PMM,AS,BorockyHenk,Ruben,VS}, it is shown in
\cite{DonohoPol}, that if $A$ is a random $m\times n$
ortho-projector matrix then with overwhelming probability $AC_p^n$ is centrally $k$-neighborly (as usual, under overwhelming probability we in this paper assume
a probability that is no more than a number exponentially decaying in $n$ away from $1$). Miraculously, \cite{DonohoPol,DonohoUnsigned} provided a precise characterization of $m$ and $k$ (in a large dimensional and statistically typical context) for which this happens. In a series of our own work (see, e.g. \cite{StojnicICASSP09,StojnicCSetam09,StojnicUpper10}) we then created an alternative probabilistic approach which was capable of matching the statistically typical results of Donoho \cite{DonohoPol} through a purely probabilistic approach.

Of course, there are many other algorithms that can be used to attack (\ref{eq:l0}). Among them are also numerous variations of the standard $\ell_1$-optimization from e.g. \cite{CWBreweighted,SChretien08,SaZh08,StojnicICASSP10knownsupp} as well as many other conceptually completely different ones from e.g. \cite{JATGomp,JAT,NeVe07,DTDSomp,NT08,DaiMil08,DonMalMon09}. While all of them are fairly successful in their own way and with respect to various types of performance measure, one of them, namely the so called AMP from \cite{DonMalMon09}, is of particular interest when it comes to $\ell_1$. What is fascinating about AMP is that it is a fairly fast algorithm (it does require a bit of tuning though) and it has provably the same statistical performance as (\ref{eq:l1}) (for more details on this see, e.g. \cite{DonMalMon09,BayMon10}). Since our main goal in this paper is to a large degree related to $\ell_1$ we stop short of reviewing further various alternatives to (\ref{eq:l1}) and instead refer to any of the above mentioned papers as well as our own \cite{StojnicCSetam09,StojnicUpper10} where these alternatives were revisited in a bit more detail.

In the rest of this paper we however look at a natural modification of $\ell_1$ called $\ell_q, 0\leq q\leq 1$.

%%%%%%%%%%%%%%%%%%%%%%%%%%%%%%%%%%%%%%%%%%%%%%%%%%%%%%%%%%%%%%%%%
\section{$\ell_q$-minimization}
\label{sec:lqmin}
%%%%%%%%%%%%%%%%%%%%%%%%%%%%%%%%%%%%%%%%%%%%%%%%%%%%%%%%%%%%%%%%%

As mentioned above, the first relaxation of (\ref{eq:l0}) that is typically employed is
the $\ell_1$ minimization from (\ref{eq:l1}). The reason for that is that it is the first of the norm relaxations that results in an optimization problem that is solvable in polynomial time. One can alternatively look at the following (tighter) relaxation (considered in e.g. \cite{GN03,GN04,GN07,FL08})
\begin{eqnarray}
\mbox{min} & & \|\x\|_{q}\nonumber \\
\mbox{subject to} & & A\x=\y, \label{eq:lq}
\end{eqnarray}
where for concreteness we assume $q\in[0,1]$ (also we assume that $q$ is a constant independent of problem dimension $n$). The optimization problem in (\ref{eq:lq}) looks very similar to the one in (\ref{eq:l1}). However, there is one important difference, the problem in (\ref{eq:l1}) is essentially a linear program and easily solvable in polynomial time. On the other hand the problem in (\ref{eq:lq}) is not known to be solvable in polynomial time. In fact it can be a very hard problem to solve. Since our goal in this paper will not be the design of algorithms that can solve (\ref{eq:lq}) quickly we refrain from a further discussion in that direction. Instead, we will assume that (\ref{eq:lq}) somehow can be solved and then we will look at scenarios when such a solution matches $\tilde{\x}$. In a way our analysis will then be useful in providing some sort of answers to the following question: if one can solve (\ref{eq:lq}) in a reasonable (if not polynomial) amount of time how likely is that its solution will be $\tilde{\x}$.

This is almost no different from the same type of question we considered when discussing performance of (\ref{eq:l1}) above and obviously the same type of question attacked in \cite{CRT,DOnoho06CS,DonohoPol,StojnicCSetam09,StojnicUpper10}. To be a bit more specific, one can then ask for what system dimensions (\ref{eq:lq}) actually works well and finds exactly the same solution as (\ref{eq:l0}), i.e. $\tilde{\x}$. A typical way to attack such a question would be to translate the results that relate to $\ell_1$ to general $\ell_q$ case. In fact that is exactly what has been done for many techniques, including obviously the RIP one developed in \cite{CRT}. Also, in our recent work \cite{StojnicLqThrBnds10} we attempted to proceed along the same lines and translate our own results from \cite{StojnicCSetam09} that relate to $\ell_1$ optimization to the case of interest here, i.e. to the $\ell_q,0\leq q\leq 1$, optimization. To provide a more detailed explanation as to what was done in \cite{StojnicLqThrBnds10} we will first recall on a couple of definitions. These definitions relate to what is known as $\ell_q,0\leq q\leq 1$, optimization thresholds.

First, we start by recalling that when one speaks about equivalence of (\ref{eq:lq}) and (\ref{eq:l0}) one actually may want to consider several types of such an equivalence. The classification into several types is roughly speaking based on the fact that the equivalence is achieved all the time, i.e. for any $\tilde{\x}$ or only sometimes, i.e. only for some $\tilde{\x}$. Since we will heavily use these concepts in the rest of the paper, we below make all of them mathematically precise (many of the definitions that we use below can be found in various forms in e.g. \cite{DonohoPol,DT,DTciss,DTjams2010,StojnicCSetam09,StojnicICASSP09,StojnicLqThrBnds10}).

We start with a well known statement (this statement in case of $\ell_1$ optimization follows directly from seminal works \cite{CRT,DOnoho06CS}). For any given constant $\alpha\leq 1$ there is a maximum
allowable value of $\beta$ such that for \emph{all} $k$-sparse $\tilde{\x}$ in (\ref{eq:yrepsystem}) the solution of (\ref{eq:lq})
is with overwhelming probability exactly the corresponding $k$-sparse $\tilde{\x}$. One can then (as is typically done) refer to this maximum allowable value of
$\beta$ as the \emph{strong threshold} (see
\cite{DonohoPol}) and denote it as $\beta_{str}^{(q)}$. Similarly, for any given constant
$\alpha\leq 1$ and \emph{all} $k$-sparse $\tilde{\x}$ with a given fixed location of non-zero components there will be a maximum allowable value of $\beta$ such that
(\ref{eq:lq}) finds the corresponding $\tilde{\x}$ in (\ref{eq:yrepsystem}) with overwhelming
probability. One can refer to this maximum allowable value of
$\beta$ as the \emph{sectional threshold} and denote it by $\beta_{sec}^{(q)}$ (more on this or similar corresponding $\ell_1$ optimization sectional thresholds definitions can be found in e.g. \cite{DonohoPol,StojnicCSetam09,StojnicLqThrBnds10}). One can also go a step further and consider scenario where for any given constant
$\alpha\leq 1$ and \emph{a} given $\tilde{\x}$
there will be a maximum allowable value of $\beta$ such that
(\ref{eq:lq}) finds that given $\tilde{\x}$ in (\ref{eq:yrepsystem}) with overwhelming
probability. One can then refer to such a $\beta$ as the \emph{weak threshold} and denote it by $\beta_{weak}^{(q)}$ (more on this and similar definitions of the weak threshold the interested reader can find in e.g. \cite{StojnicICASSP09,StojnicCSetam09,StojnicLqThrBnds10}).

When viewed within this frame the results of \cite{CRT,DOnoho06CS} established that $\ell_1$-minimization achieves recovery through a linear scaling of all important dimensions ($k$, $m$, and $n$). Moreover, for all $\beta$'s defined above lower bounds were provided in \cite{CRT}. On the other hand, the results of \cite{DonohoPol,DonohoUnsigned} established the exact values of $\beta_w^{(1)}$ and provided lower bounds on $\beta_{str}^{(1)}$ and $\beta_{sec}^{(1)}$. Our own results from \cite{StojnicCSetam09,StojnicUpper10} also established the exact values of $\beta_w^{(1)}$ and provided a different set of lower bounds on $\beta_{str}^{(1)}$ and $\beta_{sec}^{(1)}$. When it comes to a general $0\leq q\leq 1$ case, results from \cite{StojnicLqThrBnds10} established lower bounds on all three types of thresholds, $\beta_{str}^{(q)}$, $\beta_{sec}^{(q)}$, and $\beta_{weak}^{(q)}$. While establishing these bounds was an important step in the analysis of $\ell_q$ optimization, they were not fully successful all the time (on occasion, they actually fell even below the known $\ell_1$ lower bounds). In this paper we provide a substantial conceptual improvement of the results we presented in \cite{StojnicLqThrBnds10}. Such an improvement is in first place due to a recent progress we made in studying various other combinatorial problems, especially the introductory ones appearing in \cite{StojnicMoreSophHopBnds10,StojnicLiftStrSec13}. Moreover, it often leads to a substantial practical improvement as well and one may say seemingly neutralizes the deficiencies of the methods of \cite{StojnicLqThrBnds10}.

We organize the rest of the paper in the following way. In Section
\ref{sec:secthr} we present the core of the mechanism and how it can be used to obtain the sectional thresholds for $\ell_q$ minimization. In Section \ref{sec:strthr} we will then present a neat modification of the mechanism so that it can handle the strong thresholds as well. In Section \ref{sec:weakthr} we present the weak thresholds results. In Section \ref{sec:conc} we discuss obtained results and provide several conclusions related to their importance.

%%%%%%%%%%%%%%%%%%%%%%%%%%%%%%%%%%%%%%%%%%%%%%%%%%%%%%%%%%%%%%%%%
\section{Lifting $\ell_q$-minimization sectional threshold}
\label{sec:secthr}
%%%%%%%%%%%%%%%%%%%%%%%%%%%%%%%%%%%%%%%%%%%%%%%%%%%%%%%%%%%%%%%%%

In this section we start assessing the performance of $\ell_q$ minimization by looking at its sectional thresholds. Essentially, we will present a mechanism that conceptually substantially improves on results from \cite{StojnicLqThrBnds10}. We will split the presentation into two main parts, the first one that deals with the basic results needed for our analysis and the second one that deals with the core arguments.

%%%%%%%%%%%%%%%%%%%%%%%%%%%%%%%%%%%%%%%%%%%%%%%%%%%%%%%%%%%%%%%%%
\subsection{Sectional threshold preliminaries}
\label{sec:secthrprelim}
%%%%%%%%%%%%%%%%%%%%%%%%%%%%%%%%%%%%%%%%%%%%%%%%%%%%%%%%%%%%%%%%%

Below we recall on a way to quantify behavior of $\beta_{sec}^{(q)}$. In doing so we will rely on some of the mechanisms presented in \cite{StojnicCSetam09,StojnicLqThrBnds10}. Along the same lines we will assume a substantial level of familiarity with many of the well-known results that relate to the performance characterization of (\ref{eq:l1}) as well as with those presented in \cite{StojnicLqThrBnds10} that relate to $\ell_q,0\leq q\leq 1$ (we will fairly often recall on many results/definitions that we established in \cite{StojnicCSetam09,StojnicLqThrBnds10}). We start by introducing a nice way of characterizing sectional success/failure of (\ref{eq:lq}).

\begin{theorem}(Nonzero part of $\x$ has fixed location)
Assume that an $m\times n$ matrix $A$ is given. Let $\tilde{X}_{sec}$ be the collection of all $k$-sparse vectors $\tilde{\x}$ in $R^n$ for which $\tilde{\x}_1=\tilde{\x}_2=\dots=\tilde{\x}_{n-k}=0$. Let $\tilde{\x}^{(i)}$ be any $k$-sparse vector from $\tilde{X}_{sec}$. Further, assume that $\y^{(i)}=A\tilde{\x}^{(i)}$ and that $\w$ is
an $n\times 1$ vector. If
\begin{equation}
(\forall \w\in \textbf{R}^n | A\w=0) \quad  \sum_{i=n-k+1}^n |\w_i|^q<\sum_{i=1}^{n-k}|\w_{i}|^q
\label{eq:thmeqgensec1}
\end{equation}
then the solution of (\ref{eq:lq}) for every pair $(\y^{(i)},A)$ is the corresponding $\tilde{\x}^{(i)}$.
%Moreover, if
%\begin{equation}
%(\exists \w\in \textbf{R}^n | A\w=0) \quad  \sum_{i=n-k+1}^n |\w_i|>\sum_{i=1}^{n-k}|\w_{i}|
%\label{eq:thmeqgensec2}
%\end{equation}
%then there will be a $k$-sparse $\x$ that satisfies (\ref{eq:system}) and is not the solution of (\ref{eq:l1}).
\label{thm:thmgensec}
\end{theorem}

\noindent \textbf{Remark:} As mentioned in \cite{StojnicLqThrBnds10}, this result is not really our own; more on similar or even the same results can be found in e.g. \cite{DH01,FN,LN,Y,XHapp,SPH,DTbern,GN03,GN04,GN07,FL08}.

We then, following the methodology of \cite{StojnicCSetam09,StojnicLqThrBnds10},
start by defining a set $\Ssec$
\begin{equation}
\Ssec=\{\w\in S^{n-1}| \quad \sum_{i=n-k+1}^n |\w_i|^q\geq \sum_{i=1}^{n-k}|\w_{i}|^q\},\label{eq:defSsec}
\end{equation}
where $S^{n-1}$ is the unit sphere in $R^n$. Then it was established in \cite{StojnicCSetam09} that the following optimization problem is of critical importance in determining the sectional threshold of $\ell_1$-minimization
\begin{equation}
\xi_{sec}=\min_{\w\in\Ssec}\|A\w\|_2,\label{eq:negham1}
\end{equation}
where $q=1$ in the definition of $\Ssec$ (the same will remain true for any $0\leq q\leq 1$).
Namely, what was established in \cite{StojnicCSetam09} is roughly the following: if $\xi_{sec}$ is positive with overwhelming probability for certain combination of $k$, $m$, and $n$ then for $\alpha=\frac{m}{n}$ one has a lower bound $\beta_{sec}=\frac{k}{n}$ on the true value of the sectional threshold with overwhelming probability. Also, the mechanisms of \cite{StojnicCSetam09} were powerful enough to establish the concentration of $\xi_{sec}$. This essentially means that if we can show that $E\xi_{sec}>0$ for certain $k$, $m$, and $n$ we can then obtain a lower bound on the sectional threshold. In fact, this is precisely what was done in \cite{StojnicCSetam09}. However, the results we obtained for the sectional threshold through such a consideration were not exact. The main reason of course was inability to determine $E\xi_{sec}$ exactly. Instead we resorted to its lower bounds and those turned out to be loose. In \cite{StojnicLiftStrSec13} we used some of the ideas we recently introduced in \cite{StojnicMoreSophHopBnds10} to provide a substantial conceptual improvement in these bounds which in turn reflected in a conceptual improvement of the sectional thresholds (and later on an even substantial practical improvement of all strong thresholds). When it comes to general $q$ we then in \cite{StojnicLqThrBnds10} adopted the strategy similar to the one employed in \cite{StojnicCSetam09}. Again, the results we obtained for the sectional threshold through such a consideration were not exact. The main reason of course was again an inability to determine $E\xi_{sec}$ exactly and essentially the lower bounds we resorted to again turned out to be loose. In this paper we will use some of the ideas from \cite{StojnicMoreSophHopBnds10,StojnicLiftStrSec13} to provide a substantial conceptual improvement in these bounds which in turn will reflect in a conceptual (and practical) improvement of the sectional thresholds.

Below we present a way to create a lower-bound on the optimal value of (\ref{eq:negham1}).

%%%%%%%%%%%%%%%%%%%%%%%%%%%%%%%%%%%%%%%%%%%%%%%%%%%%%%%%%%%%%%%%%
\subsection{Lower-bounding $\xi_{sec}$}
\label{sec:lbxisec}
%%%%%%%%%%%%%%%%%%%%%%%%%%%%%%%%%%%%%%%%%%%%%%%%%%%%%%%%%%%%%%%%%

In this section we will look at the problem from (\ref{eq:negham1}). As mentioned earlier, we will consider a statistical scenario and assume that the elements of $A$ are i.i.d. standard normal random variables. Such a scenario was considered in \cite{StojnicLiftStrSec13} as well and the following was done.
First we reformulated the problem in (\ref{eq:negham1}) in the following way
\begin{equation}
\xi_{sec}=\min_{\w\in\Ssec}\max_{\|\y\|_2=1}\y^TA\w.\label{eq:sqrtnegham2}
\end{equation}
Then using results of \cite{StojnicHopBnds10} we established a lemma very similar to the following one:
\begin{lemma}
Let $A$ be an $m\times n$ matrix with i.i.d. standard normal components. Let $\g$ and $\h$ be $n\times 1$ and $m\times 1$ vectors, respectively, with i.i.d. standard normal components. Also, let $g$ be a standard normal random variable and let $c_3$ be a positive constant. Then
\begin{equation}
E(\max_{\w\in\Ssec}\min_{\|\y\|_2=1}e^{-c_3(\y^T A\w + g)})\leq E(\max_{\w\in\Ssec}\min_{\|\y\|_2=1}e^{-c_3(\g^T\y+\h^T\w)}).\label{eq:negexplemma}
\end{equation}\label{lemma:negexplemma}
\end{lemma}
\begin{proof}
As mentioned in \cite{StojnicLiftStrSec13} (and earlier in \cite{StojnicHopBnds10}), the proof is a standard/direct application of a theorem from \cite{Gordon85}. We will omit the details since they are pretty much the same as the those in the proof of the corresponding lemmas in \cite{StojnicHopBnds10,StojnicLiftStrSec13}. However, we do mention that the only difference between this lemma and the ones in \cite{StojnicHopBnds10,StojnicLiftStrSec13} is in set $\Ssec$. What is here $\Ssec$ it is a hypercube subset of $S^{n-1}$ in the corresponding lemma in \cite{StojnicHopBnds10} and the same set $\Ssec$ with $q=1$ in \cite{StojnicLiftStrSec13}. However, such a difference would introduce no structural changes in the proof.
\end{proof}

Following step by step what was done after Lemma 3 in \cite{StojnicHopBnds10} one arrives at the following analogue of \cite{StojnicHopBnds10}'s equation $(57)$:
\begin{equation}
E(\min_{\w\in\Ssec}\|A\w\|_2)\geq
\frac{c_3}{2}-\frac{1}{c_3}\log(E(\max_{\w\in\Ssec}(e^{-c_3\h^T\w})))
-\frac{1}{c_3}\log(E(\min_{\|\y\|_2=1}(e^{-c_3\g^T\y}))).\label{eq:chneg8}
\end{equation}
Let $c_3=c_3^{(s)}\sqrt{n}$ where $c_3^{(s)}$ is a constant independent of $n$. Then (\ref{eq:chneg8}) becomes
\begin{eqnarray}
\hspace{-.5in}\frac{E(\min_{\w\in\Ssec}\|A\w\|_2)}{\sqrt{n}}
& \geq &
\frac{c_3^{(s)}}{2}-\frac{1}{nc_3^{(s)}}\log(E(\max_{\w\in\Ssec}(e^{-c_3^{(s)}\h^T\w})))
-\frac{1}{nc_3^{(s)}}\log(E(\min_{\|\y\|_2=1}(e^{-c_3^{(s)}\sqrt{n}\g^T\y})))\nonumber \\
& = &-(-\frac{c_3^{(s)}}{2}+I_{sec}(c_3^{(s)},\beta)+I_{sph}(c_3^{(s)},\alpha)),\label{eq:chneg9}
\end{eqnarray}
where
\begin{eqnarray}
I_{sec}(c_3^{(s)},\beta) & = & \frac{1}{nc_3^{(s)}}\log(E(\max_{\w\in\Ssec}(e^{-c_3^{(s)}\h^T\w})))\nonumber \\
I_{sph}(c_3^{(s)},\alpha) & = & \frac{1}{nc_3^{(s)}}\log(E(\min_{\|\y\|_2=1}(e^{-c_3^{(s)}\sqrt{n}\g^T\y}))).\label{eq:defIs}
\end{eqnarray}

One should now note that the above bound is effectively correct for any positive constant $c_3^{(s)}$. The only thing that is then left to be done so that the above bound becomes operational is to estimate $I_{sec}(c_3^{(s)},\beta)$ and $I_{sph}(c_3^{(s)},\alpha)$.

We start with $I_{sph}(c_3^{(s)},\alpha)$. Setting
\begin{equation}
\widehat{\gamma_{sph}^{(s)}}=\frac{2c_3^{(s)}-\sqrt{4(c_3^{(s)})^2+16\alpha}}{8},\label{eq:gamaiden3}
\end{equation}
and using results of \cite{StojnicLiftStrSec13} one has
\begin{equation}
I_{sph}(c_3^{(s)},\alpha)=\frac{1}{nc_3^{(s)}}\log(Ee^{-c_3^{(s)}\sqrt{n}\|\g\|_2})\doteq
\left ( \widehat{\gamma_{sph}^{(s)}}-\frac{\alpha}{2c_3^{(s)}}\log(1-\frac{c_3^{(s)}}{2\widehat{\gamma_{sph}^{(s)}}}\right ),\label{eq:Isph}
\end{equation}
where $\doteq$ stands for an equality in the limit $n\rightarrow\infty$.

We now switch to $I_{sec}(c_3^{(s)},\beta)$. Similarly to what was stated in \cite{StojnicLiftStrSec13}, pretty good estimates for this quantity can be obtained for any $n$. However, to facilitate the exposition we will focus only on the large $n$ scenario. Let $f(\w)=-\h^T\w$. In \cite{StojnicLqThrBnds10} the following was shown
\begin{equation}
\max_{\w\in\Ssec}f(\w)=-\min_{\w\in\Ssec} -\h^T\w
\leq\min_{\gamma_{sec}\geq 0,\nu_{sec}\geq 0} f_1(q,\h,\nu_{sec},\gamma_{sec},\beta)+\gamma_{sec},\label{eq:seceq1}
%=\max_{\gamma_{sec}\geq 0}(-\frac{\|\g\|_2^2}{4\gamma_{sph}}-\gamma_{sph})
\end{equation}
where
\begin{equation}
f_1(q,\h,\nu_{sec},\gamma_{sec},\beta)=\max_{\w}\left (\sum_{i=n-k+1}^{n}(|\h_i||\w_i|+\nu_{sec}|\w_i|^q-\gamma_{sec}\w_i^2)
+\sum_{i=1}^{n-k}(|\h_i||\w_i|-\nu_{sec}|\w_i|^q-\gamma_{sec}\w_i^2)\right ).\label{eq:deff1}
\end{equation}
Then
\begin{multline}
\hspace{-.3in}I_{sec}(c_3^{(s)},\beta)  =  \frac{1}{nc_3^{(s)}}\log(E(\max_{\w\in\Ssec}(e^{-c_3^{(s)}\h^T\w}))) = \frac{1}{nc_3^{(s)}}\log(E(\max_{\w\in\Ssec}(e^{c_3^{(s)}f(\w))})))\\=\frac{1}{nc_3^{(s)}}\log(Ee^{c_3^{(s)}\sqrt{n}\min_{\gamma_{sec},\nu_{sec}\geq 0}(f_1(\h,\nu_{sec},\gamma_{sec},\beta)+\gamma_{sec})})
\doteq \frac{1}{nc_3^{(s)}}\min_{\gamma_{sec},\nu_{sec}\geq 0}\log(Ee^{c_3^{(s)}\sqrt{n}(f_1(q,\h,\nu_{sec},\gamma_{sec},\beta)+\gamma_{sec})})\\
=\min_{\gamma_{sec},\nu_{sec}\geq 0}(\frac{\gamma_{sec}}{\sqrt{n}}+\frac{1}{nc_3^{(s)}}\log(Ee^{c_3^{(s)}\sqrt{n}(f_1(q,\h,\nu_{sec},\gamma_{sec},\beta))})),\label{eq:gamaiden1sec}
\end{multline}
where, as earlier, $\doteq$ stands for equality when $n\rightarrow \infty$ and, as mentioned in \cite{StojnicLiftStrSec13}, would be obtained through the mechanism presented in \cite{SPH} (for our needs here though, even just replacing $\doteq$ with a simple $\leq$ inequality suffices). Now if one sets $\w_{i}=\frac{\w_{i}^{(s)}}{\sqrt{n}}$, $\gamma_{sec}=\gamma_{sec}^{(s)}\sqrt{n}$, and $\nu_{sec}=\nu_{sec}^{(s)}\sqrt{n}^{q-1}$ (where $\w_{i}^{(s)}$, $\gamma_{sec}^{(s)}$, and $\nu_{sec}^{(s)}$ are independent of $n$) then (\ref{eq:gamaiden1sec}) gives
\begin{multline}
I_{sec}(c_3^{(s)},\beta)
=\min_{\gamma_{sec},\nu_{sec}\geq 0}(\frac{\gamma_{sec}}{\sqrt{n}}+\frac{1}{nc_3^{(s)}}\log(Ee^{c_3^{(s)}\sqrt{n}(f_1(q,\h,\nu_{sec},\gamma_{sec},\beta))})\\
=\min_{\gamma_{sec}^{(s)},\nu_{sec}^{(s)}\geq 0}(\gamma_{sec}^{(s)}+\frac{\beta}{c_3^{(s)}}\log(Ee^{(c_3^{(s)}\max_{\w_i^{(s)}}(|\h_i||\w_i^{(s)}|+\nu_{sec}^{(s)}|\w_i^{(s)}|^q-\gamma_{sec}^{(s)}(\w_i^{(s)})^2))})
\\+\frac{1-\beta}{c_3^{(s)}}\log(Ee^{(c_3^{(s)}\max_{\w_j^{(s)}}(|\h_i||\w_j^{(s)}|-\nu_{sec}^{(s)}|\w_j^{(s)}|^q-\gamma_{sec}^{(s)}(\w_j^{(s)})^2))}))
=\min_{\gamma_{sec}^{(s)},\nu_{sec}^{(s)}\geq 0}(\gamma_{sec}^{(s)}+\frac{\beta}{c_3^{(s)}}\log(I_{sec}^{(1)})
+\frac{1-\beta}{c_3^{(s)}}\log(I_{sec}^{(2)})),\\\label{eq:gamaiden2sec}
\end{multline}
where
\begin{eqnarray}
I_{sec}^{(1)} & = & Ee^{(c_3^{(s)}\max_{\w_i^{(s)}}(|\h_i||\w_i^{(s)}|+\nu_{sec}^{(s)}|\w_i^{(s)}|^q-\gamma_{sec}^{(s)}(\w_i^{(s)})^2))}\nonumber \\
I_{sec}^{(2)} & = & Ee^{(c_3^{(s)}\max_{\w_j^{(s)}}(|\h_i||\w_j^{(s)}|-\nu_{sec}^{(s)}|\w_j^{(s)}|^q-\gamma_{sec}^{(s)}(\w_j^{(s)})^2))}.\label{eq:defI1I2sec}
\end{eqnarray}

We summarize the above results related to the sectional threshold ($\beta_{sec}^{(q)}$) in the following theorem.

\begin{theorem}(Sectional threshold - lifted lower bound)
Let $A$ be an $m\times n$ measurement matrix in (\ref{eq:system})
with i.i.d. standard normal components. Let $\tilde{X}_{sec}$ be the collection of all $k$-sparse vectors $\tilde{\x}$ in $R^n$ for which $\tilde{\x}_1=0,\tilde{\x}_2=0,,\dots,\tilde{\x}_{n-k}=0$. Let $\tilde{\x}^{(i)}$ be any $k$-sparse vector from $\tilde{X}_{sec}$. Further, assume that $\y^{(i)}=A\tilde{\x}^{(i)}$. Let $k,m,n$ be large
and let $\alpha=\frac{m}{n}$ and $\betasec^{(q)}=\frac{k}{n}$ be constants
independent of $m$ and $n$. Let $c_3^{(s)}$ be a positive constant and set
\begin{equation}
\widehat{\gamma_{sph}^{(s)}}=\frac{2c_3^{(s)}-\sqrt{4(c_3^{(s)})^2+16\alpha}}{8},\label{eq:gamaiden3thm}
\end{equation}
and
\begin{equation}
I_{sph}(c_3^{(s)},\alpha)=
\left ( \widehat{\gamma_{sph}^{(s)}}-\frac{\alpha}{2c_3^{(s)}}\log(1-\frac{c_3^{(s)}}{2\widehat{\gamma_{sph}^{(s)}}}\right ).\label{eq:Isphthm}
\end{equation}
Further let
\begin{eqnarray}
I_{sec}^{(1)} & = & Ee^{c_3^{(s)}\max_{\w_i}(|\h_i||\w_i^{(s)}|+\nu_{sec}^{(s)}|\w_i^{(s)}|^q-\gamma_{sec}^{(s)}(\w_i^{(s)})^2)}\nonumber \\
I_{sec}^{(2)} & = & Ee^{c_3^{(s)}\max_{\w_j}(|\h_j||\w_j^{(s)}|-\nu_{sec}^{(s)}|\w_j^{(s)}|^q-\gamma_{sec}^{(s)}(\w_j^{(s)})^2)}.\label{eq:defI1I2secthm}
\end{eqnarray}
and
\begin{equation}
I_{sec}(c_3^{(s)},\betasec^{(q)})=\min_{\gamma_{sec}^{(s)},\nu_{sec}^{(s)}\geq 0}(\gamma_{sec}^{(s)}+\frac{\betasec^{(q)}}{c_3^{(s)}}\log(I_{sec}^{(1)})
+\frac{1-\betasec^{(q)}}{c_3^{(s)}}\log(I_{sec}^{(2)})).\label{eq:seccondthmsec}
\end{equation}
If $\alpha$ and $\betasec^{(q)}$ are such that
\begin{equation}
\min_{c_3^{(s)}}(-\frac{c_3^{(s)}}{2}+I_{sec}(c_3^{(s)},\betasec^{(q)})+I_{sph}(c_3^{(s)},\alpha))<0,\label{eq:seccondthmsec}
\end{equation}
then with overwhelming probability the solution of (\ref{eq:lq}) for every pair $(\y^{(i)},A)$ is the corresponding $\tilde{\x}^{(i)}$.\label{thm:thmsecthrlq}
\end{theorem}
\begin{proof}
Follows from the above discussion.
\end{proof}

One also has immediately the following corollary.

\begin{corollary}(Sectional threshold - lower bound \cite{StojnicLqThrBnds10})
Let $A$ be an $m\times n$ measurement matrix in (\ref{eq:system})
with i.i.d. standard normal components. Let $\tilde{X}_{sec}$ be the collection of all $k$-sparse vectors $\tilde{\x}$ in $R^n$ for which $\tilde{\x}_1=0,\tilde{\x}_2=0,,\dots,\tilde{\x}_{n-k}=0$. Let $\tilde{\x}^{(i)}$ be any $k$-sparse vector from $\tilde{X}_{sec}$. Further, assume that $\y^{(i)}=A\tilde{\x}^{(i)}$. Let $k,m,n$ be large
and let $\alpha=\frac{m}{n}$ and $\betasec^{(q)}=\frac{k}{n}$ be constants
independent of $m$ and $n$. Let
\begin{equation}
I_{sph}(\alpha)=
-\sqrt{\alpha}.\label{eq:Isphcor}
\end{equation}
Further let
\begin{eqnarray}
I_{sec}^{(1)} & = & E\max_{\w_i}(|\h_i||\w_i^{(s)}|+\nu_{sec}^{(s)}|\w_i^{(s)}|^q-\gamma_{sec}^{(s)}(\w_i^{(s)})^2)\nonumber \\
I_{sec}^{(2)} & = & E\max_{\w_j}(|\h_j||\w_j^{(s)}|-\nu_{sec}^{(s)}|\w_j^{(s)}|^q-\gamma_{sec}^{(s)}(\w_j^{(s)})^2).\label{eq:defI1I2seccor}
\end{eqnarray}
and
\begin{equation}
I_{sec}(\betasec^{(q)})=\min_{\gamma_{sec}^{(s)},\nu_{sec}^{(s)}\geq 0}(\gamma_{sec}^{(s)}+\betasec^{(q)}I_{sec}^{(1)}
+(1-\betasec^{(q)})I_{sec}^{(2)}).\label{eq:seccondcorsec}
\end{equation}
If $\alpha$ and $\betasec^{(q)}$ are such that
\begin{equation}
I_{sec}(\betasec^{(q)})+I_{sph}(\alpha)<0,\label{eq:seccondcorsec}
\end{equation}
then with overwhelming probability the solution of (\ref{eq:lq}) for every pair $(\y^{(i)},A)$ is the corresponding $\tilde{\x}^{(i)}$.\label{cor:corsecthrlq}
\end{corollary}
\begin{proof}
Follows from the above theorem by taking $c_3^{(s)}\rightarrow 0$.
\end{proof}

The results for the sectional threshold obtained from the above theorem
are presented in Figure \ref{fig:sec}. To be a bit more specific, we selected four different values of $q$, namely $q\in\{0,0.1,0.3,0.5\}$ in addition to standard $q=1$ case already discussed in \cite{StojnicCSetam09}. Also, we present in Figure \ref{fig:sec} the results one can get from Theorem \ref{thm:thmsecthrlq} when $c_3^{(s)}\rightarrow 0$ (i.e. from Corollary \ref{cor:corsecthrlq}, see e.g. \cite{StojnicLqThrBnds10}).
\begin{figure}[htb]
\begin{minipage}[b]{.5\linewidth}
\centering
\centerline{\epsfig{figure=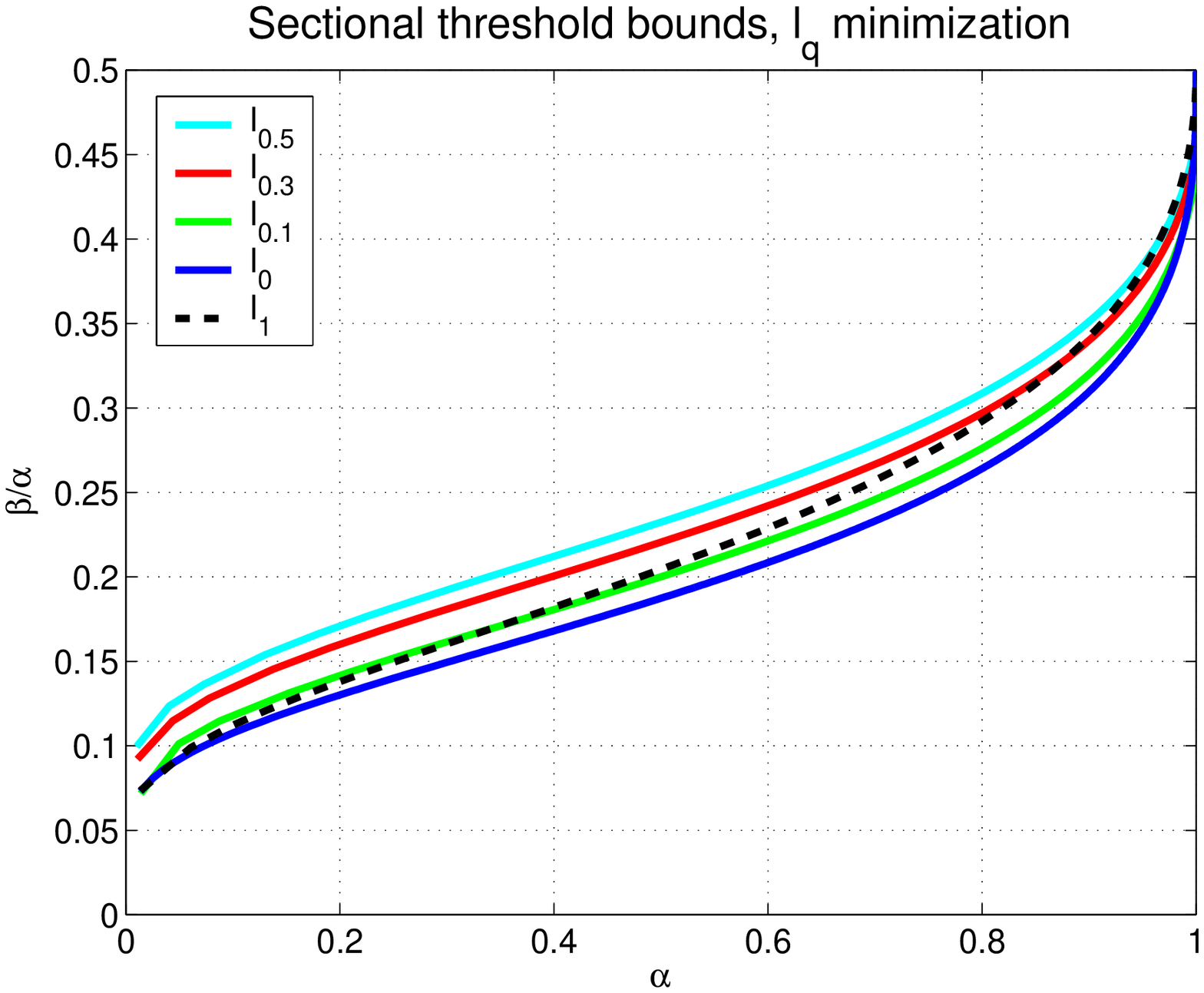,width=8cm,height=6.5cm}}
\end{minipage}
\begin{minipage}[b]{.5\linewidth}
\centering
\centerline{\epsfig{figure=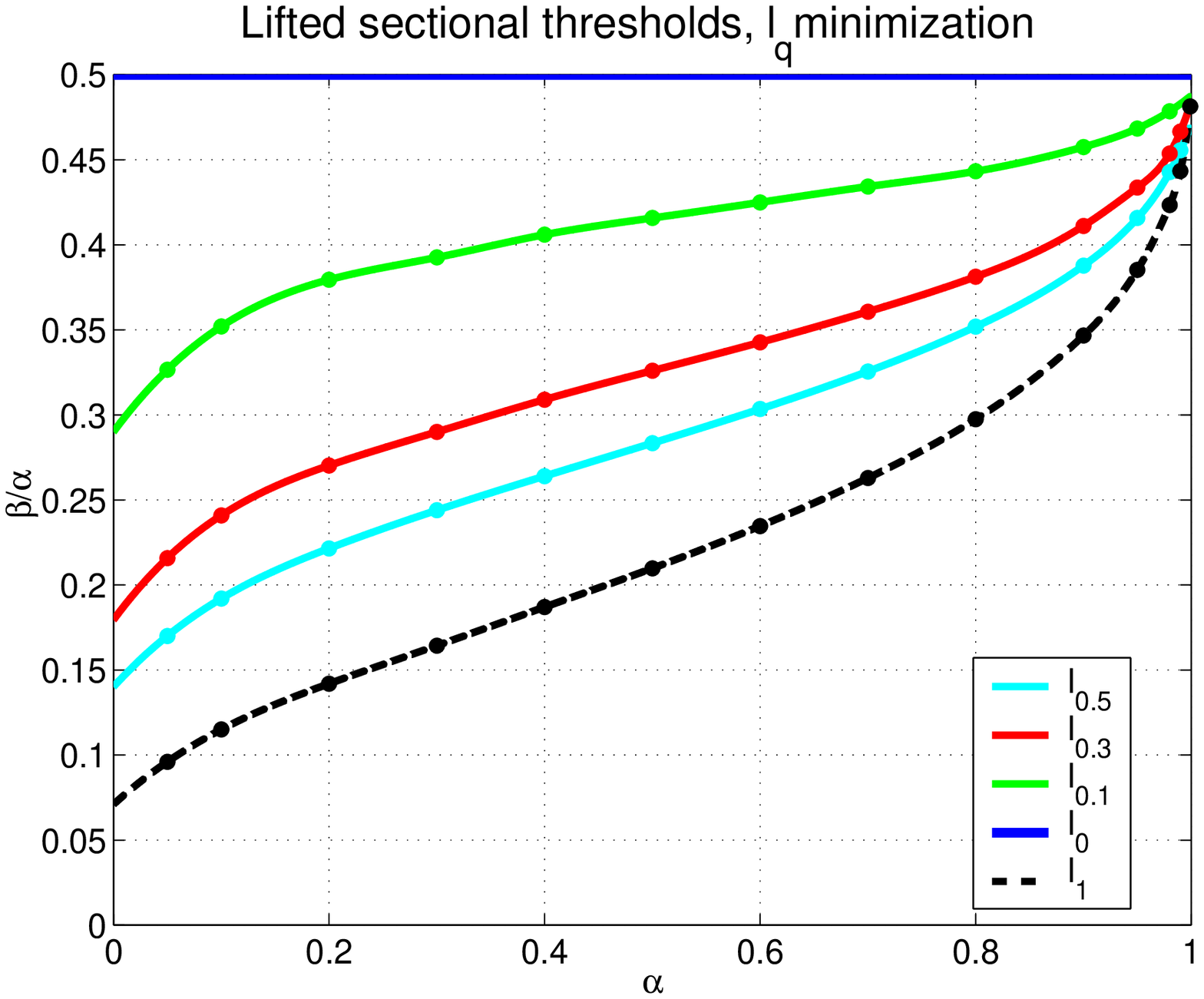,width=8cm,height=6.5cm}}
\end{minipage}
\caption{\emph{Sectional} thresholds, $\ell_q$ optimization; a) left -- $c_3\rightarrow 0$; b) right -- optimized $c_3$}
\label{fig:sec}
\end{figure}

As can be seen from Figure \ref{fig:sec}, the results for selected values of $q$ are better than for $q=1$. Also the results improve on those presented in \cite{StojnicLqThrBnds10} and essentially obtained based on Corollary \ref{cor:corsecthrlq}, i.e. Theorem \ref{thm:thmsecthrlq} for $c_3^{(s)}\rightarrow 0$.

Also, we should preface all of our discussion of presented results by emphasizing that all results are obtained after numerical computations. These are on occasion quite involved and could be imprecise. When viewed in that way one should take the results presented in Figure \ref{fig:sec} more as an illustration rather than as an exact plot of the achievable thresholds. Obtaining the presented results included several numerical optimizations which were all (except maximization over $\w$) done on a local optimum level. We do not know how (if in any way) solving them on a global optimum level would affect the location of the plotted curves. Also, additional numerical integrations were done on a finite precision level which could have potentially harmed the final results as well. Still, we believe that the methodology can not achieve substantially more than what we presented in Figure \ref{fig:sec} (and hopefully is not severely degraded with numerical integrations and maximization over $\w$). Of course, we do reemphasize that the results presented in the above theorem are completely rigorous,
it is just that some of the numerical work that we performed could have been a bit imprecise
(we firmly believe that this is not the case; however with finite numerical precision one has to be cautious all the time).

%%%%%%%%%%%%%%%%%%%%%%%%%%%%%%%%%%%%%%%%%%%%%%%%%%%%%%%%%%%%%%%%%
\subsection{Special case}
\label{sec:secthrspecial}
%%%%%%%%%%%%%%%%%%%%%%%%%%%%%%%%%%%%%%%%%%%%%%%%%%%%%%%%%%%%%%%%%

In this subsection we look at a couple of special cases that can be solved more explicitly.

%%%%%%%%%%%%%%%%%%%%%%%%%%%%%%%%%%%%%%%%%%%%%%%%%%%%%%%%%%%%%%%%%
\subsubsection{$q\rightarrow 0$}
\label{sec:secthrspecialq0}
%%%%%%%%%%%%%%%%%%%%%%%%%%%%%%%%%%%%%%%%%%%%%%%%%%%%%%%%%%%%%%%%%

We will consider case $q\rightarrow 0$. There are many methods how this particular case can be handled. Rather than obtaining the exact threshold results (which for this case is not that hard anyway), our goal here is to see what kind of performance would the methodology presented above give in this case.

We will therefore closely follow the methodology introduced above. However, we will modify certain aspects of it. To that end we start by introducing set $\Ssec^{(0)}$
\begin{equation}
\Ssec^{(0)}=\{\w^{(0)}\in S^{n-1}| \w_i^{(0)}=\w_i,n-k+1\leq i\leq n;\w_i^{(0)}=\b_i\w_i,1\leq i\leq n-k,\sum_{i=1}^{n-k}\b_i=k;\sum_{i=1}^{n}\w_i^2=1\}.\label{eq:defSsec0}
\end{equation}
It is not that hard to see that when $q\rightarrow 0$ the above set can be used to characterize sectional failure of $\ell_0$ optimization in a manner similar to the one set $\Ssec$ was used earlier to characterize sectional failure of $\ell_q$ for a general $q$. Let $f(\w^{(0)})=\h^T\w^{(0)}$ and
we start with the following line of identities
\begin{multline}
\hspace{-.5in}\max_{\w^{(0)}\in\Ssec^{(0)}}f(\w^{(0)})=-\min_{\w^{(0)}\in\Ssec^{(0)}} -\h^T\w^{(0)}\\=-\min_{\w}\max_{\gamma_{sec}\geq 0,\nu_{sec}^{(0)}\geq 0}
-\sum_{i=n-k+1}^{n}\h_i\w_i-\sum_{i=1}^{n-k}\h_i\b_i\w_i
+\nu_{sec}^{(0)}\sum_{i=1}^{n-k}\b_i
-\nu_{sec}^{(0)}k+\gamma_{sec}\sum_{i=1}^{n}\w_i^2-\gamma_{sec}\\
\leq-\max_{\gamma_{sec}\geq 0,\nu_{sec}^{(0)}\geq 0}\min_{\w}
-\sum_{i=n-k+1}^{n}\h_i\w_i-\sum_{i=1}^{n-k}\h_i\b_i\w_i
+\nu_{sec}^{(0)}\sum_{i=1}^{n-k}\b_i
-\nu_{sec}^{(0)}k+\gamma_{sec}\sum_{i=1}^{n}\w_i^2-\gamma_{sec}\\
=\min_{\gamma_{sec}\geq 0,\nu_{sec}^{(0)}\geq 0}\max_{\w}
\sum_{i=n-k+1}^{n}\h_i\w_i-\sum_{i=1}^{n-k}\h_i\b_i\w_i
-\nu_{sec}^{(0)}\sum_{i=1}^{n-k}\b_i
+\nu_{sec}^{(0)}k-\gamma_{sec}\sum_{i=1}^{n}\w_i^2+\gamma_{sec}\\
\min_{\gamma_{sec}\geq 0,\nu_{sec}^{(0)}\geq 0}
\sum_{i=n-k+1}^{n}\frac{\h_i^2}{4\gamma_{sec}}+\sum_{i=1}^{n-k}\max\{\frac{\h_i^2}{4\gamma_{sec}}-\nu_{sec}^{(0)},0\}
+\nu_{sec}^{(0)}k+\gamma_{sec}\\
=\min_{\gamma_{sec}\geq 0,\nu_{sec}^{(0)}\geq 0} f_1^{(0)}(\h,\nu_{sec},\gamma_{sec},\beta)+\gamma_{sec},\label{eq:seceq1q0}
%=\max_{\gamma_{sec}\geq 0}(-\frac{\|\g\|_2^2}{4\gamma_{sph}}-\gamma_{sph})
\end{multline}
where
\begin{equation}
f_1^{(0)}(\h,\nu_{sec},\gamma_{sec},\beta)=\left (\sum_{i=n-k+1}^{n}\frac{\h_i^2}{4\gamma_{sec}}+\sum_{i=1}^{n-k}\max\{\frac{\h_i^2}{4\gamma_{sec}}-\nu_{sec}^{(0)},0\}+\nu_{sec}^{(0)}k\right ).\label{eq:deff1q0}
\end{equation}
Now one can write analogously to (\ref{eq:gamaiden1sec})
\begin{equation}
I_{sec}^{(0)}(c_3^{(s)},\beta)
\doteq \min_{\gamma_{sec},\nu_{sec}\geq 0}(\frac{\gamma_{sec}}{\sqrt{n}}+\frac{1}{nc_3^{(s)}}\log(Ee^{c_3^{(s)}\sqrt{n}(f_1^{(0)}(\h,\nu_{sec},\gamma_{sec},\beta))})).\label{eq:gamaiden1secq0}
\end{equation}
After further introducing $\gamma_{sec}=\gamma_{sec}^{(s)}\sqrt{n}$, and $\nu_{sec}^{(0)}=\nu_{sec}^{(0,s)}\sqrt{n}^{-1}$ (where $\gamma_{sec}^{(s)}$, and $\nu_{sec}^{(0,s)}$ are independent of $n$) one can write analogously to (\ref{eq:gamaiden2sec})
\begin{multline}
I_{sec}^{(0)}(c_3^{(s)},\beta)
\doteq\min_{\gamma_{sec},\nu_{sec}\geq 0}(\frac{\gamma_{sec}}{\sqrt{n}}+\frac{1}{nc_3^{(s)}}\log(Ee^{c_3^{(s)}\sqrt{n}(f_1^{(0)}(\h,\nu_{sec},\gamma_{sec},\beta))})\\
=\min_{\gamma_{sec}^{(s)},\nu_{sec}^{(0,s)}\geq 0}(\gamma_{sec}^{(s)}+\beta\nu_{sec}^{(0,s)}+\frac{\beta}{c_3^{(s)}}\log(Ee^{(c_3^{(s)}\frac{\h_i^2}{\gamma_{sec}^{(s)}})})
+\frac{1-\beta}{c_3^{(s)}}\log(Ee^{(c_3^{(s)}\max\{\frac{\h_i^2}{\gamma_{sec}^{(s)}}-\nu_{sec}^{(0,s)},0\})}))\\
=\min_{\gamma_{sec}^{(s)},\nu_{sec}^{(s)}\geq 0}(\gamma_{sec}^{(s)}+\beta\nu_{sec}^{(0,s)}+\frac{\beta}{c_3^{(s)}}\log(I_{sec}^{(0,1)})
+\frac{1-\beta}{c_3^{(s)}}\log(I_{sec}^{(0,2)})),\\\label{eq:gamaiden2secq0}
\end{multline}
where
\begin{eqnarray}
I_{sec}^{(0,1)} & = & Ee^{(c_3^{(s)}\frac{\h_i^2}{\gamma_{sec}^{(s)}})}\nonumber \\
I_{sec}^{(0,2)} & = & Ee^{(c_3^{(s)}\max\{\frac{\h_i^2}{\gamma_{sec}^{(s)}}-\nu_{sec}^{(0,s)},0\})}.\label{eq:defI1I2secq0}
\end{eqnarray}
One can then write analogously to (\ref{eq:seccondthmsec})
\begin{equation}
\min_{c_3^{(s)}}(-\frac{c_3^{(s)}}{2}+I_{sec}^{(0)}(c_3^{(s)},\beta)+I_{sph}(c_3^{(s)},\alpha))<0.\label{eq:seccondthmsecq0}
\end{equation}
Setting $b=\frac{c_3^{(s)}}{4\gamma_{sec}}$, $\nu_{sec}^{(0,s,\gamma)}=4\gamma_{sec}\nu_{sec}^{(0,s)}$, and solving the integrals one from (\ref{eq:seccondthmsecq0}) has the following condition for $\beta$ and $\alpha$
\begin{multline}
-\beta\frac{1}{2c_3^{(s)}}\log\left (\frac{\alpha}{(c_3^{(s)})^2}\right )+\frac{b\nu_{sec}^{(0,s,\gamma)}\beta}{c_3^{(s)}}+c_3^{(s)}\frac{1-2b}{4b}\\+\frac{1}{c_3^{(s)}}\log \left (\frac{e^{-b\nu_{sec}^{(0,s,\gamma)}}}{\sqrt{1-2b}}\mbox{erfc}\left (\sqrt{\frac{1-2b}{2}\nu_{sec}^{(0,s,\gamma)}}\right )
+\mbox{erf}\left (\sqrt{\frac{\nu_{sec}^{(0,s,\gamma)}}{2}}\right )\right )+I_{sph}(c_3^{(s)},\alpha)<0.\label{q0condIsph1sec}
\end{multline}
Now, assuming $c_3^{(s)}$ is large one has
\begin{equation}
I_{sph}(c_3^{(s)},\alpha)\approx -\frac{\alpha}{2c_3^{(s)}}-\frac{\alpha}{2c_3^{(s)}}\log(1+\frac{(c_3^{(s)})^2}{\alpha}).\label{q0condIsph2sec}
\end{equation}
Setting
\begin{eqnarray}
1-2b & = & \frac{\alpha}{(c_3^{(s)})^2}\nonumber \\
\nu_{sec}^{(0,s,\gamma)} & = & \log(\frac{(c_3^{(s)})^2}{\alpha}),\label{eq:q0condsec}
\end{eqnarray}
one from (\ref{q0condIsph1sec}) and (\ref{q0condIsph2sec}) has
\begin{multline}
-\beta\frac{1}{2c_3^{(s)}}\log\left (\frac{\alpha}{(c_3^{(s)})^2}\right )+\frac{b\nu_{sec}^{(0,s,\gamma)}\beta}{c_3^{(s)}}+c_3^{(s)}\frac{1-2b}{4b}\\+\frac{1}{c_3^{(s)}}\log \left (\frac{e^{-b\nu_{sec}^{(0,s,\gamma)}}}{\sqrt{1-2b}}\mbox{erfc}\left (\sqrt{\frac{1-2b}{2}\nu_{sec}^{(0,s,\gamma)}}\right )
+\mbox{erf}\left (\sqrt{\frac{\nu_{sec}^{(0,s,\gamma)}}{2}}\right )\right )+I_{sph}(c_3^{(s)},\alpha)
= O\left (\frac{(2\beta-\alpha)\log(c_3^{(s)})}{c_3^{(s)}}\right ).\\\label{q0condIsph3sec}
\end{multline}
Then from (\ref{q0condIsph3sec}) one has that as long as $\beta_{sec}^{(0)}<\frac{\alpha}{2}-\epsilon_{\ell_0}$, where $\epsilon_{\ell_0}$ is a small positive constant (adjusted with respect to $c_{3}^{(s)}$) (\ref{eq:seccondthmsecq0}) holds which is, as stated above, an analogue to the condition given in (\ref{eq:seccondthmsec}) in Theorem \ref{thm:thmsecthrlq}. This essentially means that when $q\rightarrow 0$ one has that the threshold curve approaches the best possible curve $\frac{\beta}{\alpha}=\frac{1}{2}$. While, as we stated at the beginning of this subsection, this particular fact can be shown in many different ways, the way we chose to present additionally shows that the methodology of this paper is actually capable of achieving the theoretically best possible threshold curve. Of course that does not necessarily mean that the same would be true for any $q>0$. However, it may serve as an indicator that maybe even for other values of $q$ it does achieve the values that are somewhat close to the true thresholds. In that light one can believe a bit more in the numerical results we presented earlier for various different $q$'s. Of course, one still has to be careful. Namely, while we have solid indicators that the methodology is quite powerful all of what we just discussed still does not necessarily imply that the numerical results we presented earlier are completely exact. It essentially just shows that it may make sense that they provide substantially better performance guarantees than the corresponding ones obtained in Corollary \ref{cor:corstrthrlq} (and earlier in \cite{StojnicLqThrBnds10}) for $c_3^{(s)}\rightarrow 0$.

\subsubsection{$q=\frac{1}{2}$}
\label{sec:secthrspecialq05}
%%%%%%%%%%%%%%%%%%%%%%%%%%%%%%%%%%%%%%%%%%%%%%%%%%%%%%%%%%%%%%%%%

Another special case that allows a further simplification of the results presented in Theorem \ref{thm:thmsecthrlq} is when $q=\frac{1}{2}$. As discussed in \cite{StojnicLqThrBnds10}, when $q=\frac{1}{2}$ one can also be more explicit when it comes to the optimization over $\w$. Namely, taking simply the derivatives one finds
\begin{equation*}
|\h_i|\pm q\nu_{str}^{(s)}|\w_i^{(s)}|^{q-1}-2\gamma_{str}^{(s)}|\w_i^{(s)}|=0,
\end{equation*}
which when $q=\frac{1}{2}$ gives
\begin{eqnarray}
& & |\h_i|\pm\frac{1}{2}\nu_{str}^{(s)}|\w_i^{(s)}|^{-1/2}-2\gamma_{str}^{(s)}|\w_i^{(s)}|=0\nonumber \\
& \Leftrightarrow & |\h_i|\sqrt{|\w_i^{(s)}|}\pm\frac{1}{2}\nu_{str}^{(s)}-2\gamma_{str}^{(s)}\sqrt{|\w_i^{(s)}|}^{3}=0,\label{eq:cubicq05}
\end{eqnarray}
which is a cubic equation and can be solved explicitly. This of course substantially facilitates the integrations over $\h_i$. Also, similar strategy can be applied for other rational $q$. However, as mentioned in \cite{StojnicLqThrBnds10}, the ``explicit" solutions soon become more complicated than the numerical ones and we skip presenting them.

%%%%%%%%%%%%%%%%%%%%%%%%%%%%%%%%%%%%%%%%%%%%%%%%%%%%%%%%%%%%%%%%%
\section{Lifting $\ell_q$-minimization strong threshold}
\label{sec:strthr}
%%%%%%%%%%%%%%%%%%%%%%%%%%%%%%%%%%%%%%%%%%%%%%%%%%%%%%%%%%%%%%%%%

In this section we look at the so-called strong thresholds of $\ell_q$ minimization. Essentially, we will attempt to adapt the mechanism we presented in the previous section. We will again split the presentation into two main parts, the first one that deals with the basic results needed for our analysis and the second one that deals with the core arguments.

%%%%%%%%%%%%%%%%%%%%%%%%%%%%%%%%%%%%%%%%%%%%%%%%%%%%%%%%%%%%%%%%%
\subsection{Strong threshold preliminaries}
\label{sec:strthrprelim}
%%%%%%%%%%%%%%%%%%%%%%%%%%%%%%%%%%%%%%%%%%%%%%%%%%%%%%%%%%%%%%%%%

Below we start by recalling on a way to quantify behavior of $\beta_{str}^{(q)}$. In doing so we will rely on some of the mechanisms presented in \cite{StojnicCSetam09,StojnicLqThrBnds10}. As earlier, we will fairly often recall on many results/definitions that we established in \cite{StojnicCSetam09,StojnicLqThrBnds10}. We start by introducing a nice way of characterizing strong success/failure of (\ref{eq:lq}).

\begin{theorem}(Nonzero part of $\x$ has fixed location)
Assume that an $m\times n$ matrix $A$ is given. Let $\tilde{X}_{str}$ be the collection of all $k$-sparse vectors $\tilde{\x}$ in $R^n$. Let $\tilde{\x}^{(i)}$ be any $k$-sparse vector from $\tilde{X}_{str}$. Further, assume that $\y^{(i)}=A\tilde{\x}^{(i)}$ and that $\w$ is
an $n\times 1$ vector. If
\begin{equation}
(\forall \w\in \textbf{R}^n | A\w=0) \quad  \sum_{i=1}^n \b_i |\w_i|^q>0,\sum_{i=1}^n\b_i=2n-k,\b_i^2=1),\label{eq:thmeqgenstr1}
\end{equation}
then the solution of (\ref{eq:lq}) for every pair $(\y^{(i)},A)$ is the corresponding $\tilde{\x}^{(i)}$.
%Moreover, if
%\begin{equation}
%(\exists \w\in \textbf{R}^n | A\w=0) \quad  \sum_{i=n-k+1}^n |\w_i|>\sum_{i=1}^{n-k}|\w_{i}|
%\label{eq:thmeqgensec2}
%\end{equation}
%then there will be a $k$-sparse $\x$ that satisfies (\ref{eq:system}) and is not the solution of (\ref{eq:l1}).
\label{thm:thmgenstr}
\end{theorem}

\noindent \textbf{Remark:} As mentioned earlier (and in \cite{StojnicLqThrBnds10}), this result is not really our own; more on similar or even the same results can be found in e.g. \cite{DH01,FN,LN,Y,XHapp,SPH,DTbern,GN03,GN04,GN07,FL08}.

We then, following the methodology of the previous section (and ultimately of \cite{StojnicCSetam09,StojnicLqThrBnds10}),
start by defining a set $\Sstr$
\begin{equation}
\Sstr=\{\w\in S^{n-1}| \quad \sum_{i=1}^n \b_i|\w_i|^q\leq 0, \sum_{i=1}^n\b_i=2n-k,\b_i^2=1\},\label{eq:defSstr}
\end{equation}
where $S^{n-1}$ is the unit sphere in $R^n$. The methodology of the previous section (and ultimately the one of \cite{StojnicCSetam09}) then proceeds by considering the following optimization problem
\begin{equation}
\xi_{str}=\min_{\w\in\Sstr}\|A\w\|_2,\label{eq:negham1str}
\end{equation}
where $q=1$ in the definition of $\Sstr$ (the same will remain true for any $0\leq q\leq 1$). Following what was done in the previous section one roughly has the following: if $\xi_{str}$ is positive with overwhelming probability for certain combination of $k$, $m$, and $n$ then for $\alpha=\frac{m}{n}$ one has a lower bound $\beta_{str}=\frac{k}{n}$ on the true value of the strong threshold with overwhelming probability. Also, the mechanisms of \cite{StojnicCSetam09} were powerful enough to establish the concentration of $\xi_{str}$. This essentially means that if we can show that $E\xi_{str}>0$ for certain $k$, $m$, and $n$ we can then obtain a lower bound on the strong threshold. In fact, this is precisely what was done in \cite{StojnicCSetam09}. However, the results we obtained for the strong threshold through such a consideration were not exact. The main reason of course was inability to determine $E\xi_{str}$ exactly. Instead we resorted to its lower bounds and those turned out to be loose. In \cite{StojnicLiftStrSec13} we used some of the ideas we recently introduced in \cite{StojnicMoreSophHopBnds10} to provide a substantial conceptual improvement in these bounds which in turn reflected in a conceptual improvement of the sectional thresholds (and later on an even substantial practical improvement of all strong thresholds). Since our analysis from the previous section hints that such a methodology could be successful in improving the sectional thresholds even for general $q$ one can be tempted to believe that it would work even better for the strong thresholds.

When it comes to the strong thresholds for a general $q$ we actually already in \cite{StojnicLqThrBnds10} adopted the strategy similar to the one employed in \cite{StojnicCSetam09}. However, the results we obtained for the through such a consideration were again not exact. The main reason again was an inability to determine $E\xi_{str}$ exactly and essentially the lower bounds we resorted to again turned out to be loose. In this section we will use some of the ideas from the previous section (and essentially those from \cite{StojnicMoreSophHopBnds10,StojnicLiftStrSec13}) to provide a substantial conceptual improvement in these bounds. A limited numerical exploration also indicates that they in turn will reflect in practical improvement of the strong thresholds as well.

We start by emulating what was done in the previous section, i.e. by presenting a way to create a lower-bound on the optimal value of (\ref{eq:negham1str}).

%%%%%%%%%%%%%%%%%%%%%%%%%%%%%%%%%%%%%%%%%%%%%%%%%%%%%%%%%%%%%%%%%
\subsection{Lower-bounding $\xi_{str}$}
\label{sec:lbxistr}
%%%%%%%%%%%%%%%%%%%%%%%%%%%%%%%%%%%%%%%%%%%%%%%%%%%%%%%%%%%%%%%%%

In this section we will look at the problem from (\ref{eq:negham1str}). We recall that as earlier, we will consider a statistical scenario and assume that the elements of $A$ are i.i.d. standard normal random variables. Such a scenario was considered in \cite{StojnicLiftStrSec13} as well and the following was done.
First we reformulated the problem in (\ref{eq:negham1str}) in the following way
\begin{equation}
\xi_{str}=\min_{\w\in\Sstr}\max_{\|\y\|_2=1}\y^TA\w.\label{eq:sqrtnegham2str}
\end{equation}
Then using results of \cite{StojnicHopBnds10} we established a lemma very similar to the following one:
\begin{lemma}
Let $A$ be an $m\times n$ matrix with i.i.d. standard normal components. Let $\g$ and $\h$ be $n\times 1$ and $m\times 1$ vectors, respectively, with i.i.d. standard normal components. Also, let $g$ be a standard normal random variable and let $c_3$ be a positive constant. Then
\begin{equation}
E(\max_{\w\in\Sstr}\min_{\|\y\|_2=1}e^{-c_3(\y^T A\w + g)})\leq E(\max_{\w\in\Ssec}\min_{\|\y\|_2=1}e^{-c_3(\g^T\y+\h^T\w)}).\label{eq:negexplemmastr}
\end{equation}\label{lemma:negexplemmastr}
\end{lemma}
\begin{proof}
As mentioned in the previous section (as well as in \cite{StojnicLiftStrSec13} and earlier in \cite{StojnicHopBnds10}), the proof is a standard/direct application of a theorem from \cite{Gordon85}. We will again omit the details since they are pretty much the same as the those in the proof of the corresponding lemmas in \cite{StojnicHopBnds10,StojnicLiftStrSec13}. However, we do mention that the only difference between this lemma and the ones from previous section and in \cite{StojnicHopBnds10,StojnicLiftStrSec13} is in set $\Sstr$. However, such a difference would introduce no structural changes in the proof.
\end{proof}

Following step by step what was done after Lemma 3 in \cite{StojnicHopBnds10} one arrives at the following analogue of \cite{StojnicHopBnds10}'s equation $(57)$:
\begin{equation}
E(\min_{\w\in\Sstr}\|A\w\|_2)\geq
\frac{c_3}{2}-\frac{1}{c_3}\log(E(\max_{\w\in\Sstr}(e^{-c_3\h^T\w})))
-\frac{1}{c_3}\log(E(\min_{\|\y\|_2=1}(e^{-c_3\g^T\y}))).\label{eq:chneg8str}
\end{equation}
Let $c_3=c_3^{(s)}\sqrt{n}$ where $c_3^{(s)}$ is a constant independent of $n$. Then (\ref{eq:chneg8str}) becomes
\begin{eqnarray}
\hspace{-.5in}\frac{E(\min_{\w\in\Sstr}\|A\w\|_2)}{\sqrt{n}}
& \geq &
\frac{c_3^{(s)}}{2}-\frac{1}{nc_3^{(s)}}\log(E(\max_{\w\in\Sstr}(e^{-c_3^{(s)}\h^T\w})))
-\frac{1}{nc_3^{(s)}}\log(E(\min_{\|\y\|_2=1}(e^{-c_3^{(s)}\sqrt{n}\g^T\y})))\nonumber \\
& = &-(-\frac{c_3^{(s)}}{2}+I_{str}(c_3^{(s)},\beta)+I_{sph}(c_3^{(s)},\alpha)),\label{eq:chneg9str}
\end{eqnarray}
where
\begin{eqnarray}
I_{str}(c_3^{(s)},\beta) & = & \frac{1}{nc_3^{(s)}}\log(E(\max_{\w\in\Sstr}(e^{-c_3^{(s)}\h^T\w})))\nonumber \\
I_{sph}(c_3^{(s)},\alpha) & = & \frac{1}{nc_3^{(s)}}\log(E(\min_{\|\y\|_2=1}(e^{-c_3^{(s)}\sqrt{n}\g^T\y}))).\label{eq:defIsstr}
\end{eqnarray}

One should now note that the above bound is effectively correct for any positive constant $c_3^{(s)}$. The only thing that is then left to be done so that the above bound becomes operational is to estimate $I_{sec}(c_3^{(s)},\beta)$ and $I_{sph}(c_3^{(s)},\alpha)$. Of course, $I_{sph}(c_3^{(s)},\alpha)$ has already been characterized in (\ref{eq:gamaiden3}) and (\ref{eq:Isph}). That basically means that the only thing that is left to characterize is $I_{str}(c_3^{(s)},\beta)$. Similarly to what was stated in \cite{StojnicLiftStrSec13}, pretty good estimates for this quantity can be obtained for any $n$. However, to facilitate the exposition we will, as earlier, focus only on the large $n$ scenario. Let $f(\w)=-\h^T\w$. Following \cite{StojnicLqThrBnds10} one can arrive at
\begin{equation}
\max_{\w\in\Sstr}f(\w)=-\min_{\w\in\Sstr} -\h^T\w
\leq\min_{\gamma_{str}\geq 0,\nu_{str}\geq 0} f_2(q,\h,\nu_{str},\gamma_{str},\beta)+\gamma_{str},\label{eq:seceq1}
%=\max_{\gamma_{sec}\geq 0}(-\frac{\|\g\|_2^2}{4\gamma_{sph}}-\gamma_{sph})
\end{equation}
where
\begin{equation}
f_2(q,\h,\nu_{str},\gamma_{str},\beta)=\max_{\w,\b_i^2=1}\left (\sum_{i=1}^{n}(|\h_i||\w_i|-\nu_{str}^{(1)}\b_i|\w_i|^q-\gamma_{str}\w_i^2)
+\nu_{str}^{(2)}\sum_{i=1}^{n}\b_i-\nu_{str}^{(2)}(n-2k) \right ).\label{eq:deff1str}
\end{equation}
Then
\begin{multline}
I_{str}(c_3^{(s)},\beta)  =  \frac{1}{nc_3^{(s)}}\log(E(\max_{\w\in\Sstr}(e^{-c_3^{(s)}\h^T\w}))) = \frac{1}{nc_3^{(s)}}\log(E(\max_{\w\in\Sstr}(e^{c_3^{(s)}f(\w))})))\\
\hspace{-.3in}=\frac{1}{nc_3^{(s)}}\log(Ee^{c_3^{(s)}\sqrt{n}\min_{\gamma_{str},\nu_{str}^{(1)},\nu_{str}^{(2)}\geq 0}(f_2(\h,\nu_{str},\gamma_{str},\beta)+\gamma_{str})})
\doteq \frac{1}{nc_3^{(s)}}\min_{\gamma_{str},\nu_{str}\geq 0}\log(Ee^{c_3^{(s)}\sqrt{n}(f_2(q,\h,\nu_{str},\gamma_{str},\beta)+\gamma_{str})})\\
=\min_{\gamma_{str},\nu_{str}^{(1)},\nu_{str}^{(2)}\geq 0}(\frac{\gamma_{str}}{\sqrt{n}}+\frac{1}{nc_3^{(s)}}\log(Ee^{c_3^{(s)}\sqrt{n}(f_2(q,\h,\nu_{str},\gamma_{str},\beta))})),\label{eq:gamaiden1str}
\end{multline}
where, as earlier, $\doteq$ stands for equality when $n\rightarrow \infty$. Now if one sets $\w_{i}=\frac{\w_{i}^{(s)}}{\sqrt{n}}$, $\gamma_{str}=\gamma_{str}^{(s)}\sqrt{n}$, $\nu_{str}^{(1)}=\nu_{str}^{(1,s)}\sqrt{n}^{q-1}$, and $\nu_{str}^{(2)}=\nu_{str}^{(2,s)}\sqrt{n}$ (where $\w_{i}^{(s)}$, $\gamma_{str}^{(s)}$, $\nu_{str}^{(1,s)}$, and $\nu_{str}^{(2,s)}$ are independent of $n$) then (\ref{eq:gamaiden1str}) gives
\begin{multline}
I_{str}(c_3^{(s)},\beta)
=\min_{\gamma_{str},\nu_{str}^{(1)},\nu_{str}^{(2)}\geq 0}(\frac{\gamma_{str}}{\sqrt{n}}+\frac{1}{nc_3^{(s)}}\log(Ee^{c_3^{(s)}\sqrt{n}(f_2(q,\h,\nu_{str},\gamma_{str},\beta))})\\
\hspace{-.5in}=\min_{\gamma_{str}^{(s)},\nu_{str}^{(1,s)},\nu_{str}^{(2,s)}\geq 0}(\gamma_{str}^{(s)}+\nu_{str}^{(2,s)}(2\beta-1)+\frac{1}{c_3^{(s)}}\log\left (Ee^{\left (c_3^{(s)}
\max_{\w,\b_i^2=1}\left ((|\h_i||\w_i^{(s)}|-\nu_{str}^{(1,s)}\b_i|\w_i^{(s)}|^q-\gamma_{str}^{(s)}(\w_i^{(s)})^2)
+\nu_{str}^{(2,s)}\sum_{i=1}^{n}\b_i \right )\right )}\right )
\\=\min_{\gamma_{str}^{(s)},\nu_{str}^{(1,s)},\nu_{str}^{(2,s)}\geq 0}(\gamma_{str}^{(s)}+\nu_{str}^{(2,s)}(2\beta-1)+\frac{1}{c_3^{(s)}}\log(I_{str}^{(1)})),\\\label{eq:gamaiden2str}
\end{multline}
where
\begin{equation}
I_{str}^{(1)}  = Ee^{\left (c_3^{(s)}\max_{\w,\b_i^2=1}\left ((|\h_i||\w_i^{(s)}|-\nu_{str}^{(1,s)}\b_i|\w_i^{(s)}|^q-\gamma_{str}^{(s)}(\w_i^{(s)})^2)
+\nu_{str}^{(2,s)}\sum_{i=1}^{n}\b_i \right )\right )}.\label{eq:defI1I2str}
\end{equation}

We summarize the above results related to the sectional threshold ($\beta_{str}^{(q)}$) in the following theorem.

\begin{theorem}(Strong threshold - lifted lower bound)
Let $A$ be an $m\times n$ measurement matrix in (\ref{eq:system})
with i.i.d. standard normal components. Let $\tilde{X}_{str}$ be the collection of all $k$-sparse vectors $\tilde{\x}$ in $R^n$. Let $\tilde{\x}^{(i)}$ be any $k$-sparse vector from $\tilde{X}_{str}$. Further, assume that $\y^{(i)}=A\tilde{\x}^{(i)}$. Let $k,m,n$ be large
and let $\alpha=\frac{m}{n}$ and $\betastr^{(q)}=\frac{k}{n}$ be constants
independent of $m$ and $n$. Let $c_3^{(s)}$ be a positive constant and set
\begin{equation}
\widehat{\gamma_{sph}^{(s)}}=\frac{2c_3^{(s)}-\sqrt{4(c_3^{(s)})^2+16\alpha}}{8},\label{eq:gamaiden3thmstr}
\end{equation}
and
\begin{equation}
I_{sph}(c_3^{(s)},\alpha)=
\left ( \widehat{\gamma_{sph}^{(s)}}-\frac{\alpha}{2c_3^{(s)}}\log(1-\frac{c_3^{(s)}}{2\widehat{\gamma_{sph}^{(s)}}}\right ).\label{eq:Isphthmstr}
\end{equation}
Further let
\begin{equation}
I_{str}^{(1)}  = Ee^{\left (c_3^{(s)}\max_{\w,\b_i^2=1}\left ((|\h_i||\w_i^{(s)}|-\nu_{str}^{(1,s)}\b_i|\w_i^{(s)}|^q-\gamma_{str}^{(s)}(\w_i^{(s)})^2)
+\nu_{str}^{(2,s)}\sum_{i=1}^{n}\b_i \right )\right )}.\label{eq:defI1I2secthmstr}
\end{equation}
and
\begin{equation}
I_{str}(c_3^{(s)},\betastr^{(q)})=\min_{\gamma_{str}^{(s)},\nu_{str}^{(1,s)},\nu_{str}^{(2,s)}\geq 0}(\gamma_{str}^{(s)}+\nu_{str}^{(2,s)}(2\betastr^{(q)}-1)+\frac{1}{c_3^{(s)}}\log(I_{str}^{(1)})).\label{eq:seccondthmstr}
\end{equation}
If $\alpha$ and $\betastr^{(q)}$ are such that
\begin{equation}
\min_{c_3^{(s)}}(-\frac{c_3^{(s)}}{2}+I_{str}(c_3^{(s)},\betastr^{(q)})+I_{sph}(c_3^{(s)},\alpha))<0,\label{eq:seccondthmstr}
\end{equation}
then with overwhelming probability the solution of (\ref{eq:lq}) for every pair $(\y^{(i)},A)$ is the corresponding $\tilde{\x}^{(i)}$.\label{thm:thmstrthrlq}
\end{theorem}
\begin{proof}
Follows from the above discussion.
\end{proof}

One also has immediately the following corollary.

\begin{corollary}(Strong threshold - lower bound \cite{StojnicLqThrBnds10})
Let $A$ be an $m\times n$ measurement matrix in (\ref{eq:system})
with i.i.d. standard normal components. Let $\tilde{X}_{str}$ be the collection of all $k$-sparse vectors $\tilde{\x}$ in $R^n$. Let $\tilde{\x}^{(i)}$ be any $k$-sparse vector from $\tilde{X}_{str}$. Further, assume that $\y^{(i)}=A\tilde{\x}^{(i)}$. Let $k,m,n$ be large
and let $\alpha=\frac{m}{n}$ and $\betastr^{(q)}=\frac{k}{n}$ be constants
independent of $m$ and $n$. Let
\begin{equation}
I_{sph}(\alpha)=
-\sqrt{\alpha}.\label{eq:Isphcorstr}
\end{equation}
Further let
\begin{equation}
I_{str}^{(1)}=\max_{\w,\b_i^2=1}\left ((|\h_i||\w_i^{(s)}|-\nu_{str}^{(1,s)}\b_i|\w_i^{(s)}|^q-\gamma_{str}^{(s)}(\w_i^{(s)})^2)
+\nu_{str}^{(2,s)}\sum_{i=1}^{n}\b_i \right ).\label{eq:defI1I2strcor}
\end{equation}
and
\begin{equation}
I_{str}(\betastr^{(q)})=\min_{\gamma_{str}^{(s)},\nu_{str}^{(1,s)},\nu_{str}^{(2,s)}\geq 0}(\gamma_{str}^{(s)}+\nu_{str}^{(2,s)}(2\betastr^{(q)}-1)+I_{str}^{(1)}).\label{eq:seccondcorstr}
\end{equation}
If $\alpha$ and $\betastr^{(q)}$ are such that
\begin{equation}
I_{str}(\betastr^{(q)})+I_{sph}(\alpha)<0,\label{eq:seccondcorstr}
\end{equation}
then with overwhelming probability the solution of (\ref{eq:lq}) for every pair $(\y^{(i)},A)$ is the corresponding $\tilde{\x}^{(i)}$.\label{cor:corstrthrlq}
\end{corollary}
\begin{proof}
Follows from the above theorem by taking $c_3^{(s)}\rightarrow 0$.
\end{proof}

\noindent \textbf{Remark:} Although the results in the above corollary appear visually a bit different from those given in \cite{StojnicLqThrBnds10} it is not that hard to show that they are in fact the same.

%Obtaining numerical results based on Theorem \ref{thm:thmgenstr} is even harder than obtaining the corresponding sectional ones using Theorem \ref{thm:thmgensec}. One now has one extra optimization to do. While we can present results which would resemble those from Figure \ref{fig:sec} we can not be assured of their validity. We skip such a presentation and instead leave it for a detailed numerical study of all theoretical results presented in this paper that we will present in a separate paper.

The results for the strong threshold obtained from the above theorem
are presented in Figure \ref{fig:str}. To be a bit more specific, we again selected four different values of $q$, namely $q\in\{0,0.1,0.3,0.5\}$ in addition to standard $q=1$ case already discussed in \cite{StojnicCSetam09}. Also, we present in Figure \ref{fig:str} the results one can get from Theorem \ref{thm:thmstrthrlq} when $c_3^{(s)}\rightarrow 0$ (i.e. from Corollary \ref{cor:corstrthrlq}, see e.g. \cite{StojnicLqThrBnds10}).
\begin{figure}[htb]
\begin{minipage}[b]{.5\linewidth}
\centering
\centerline{\epsfig{figure=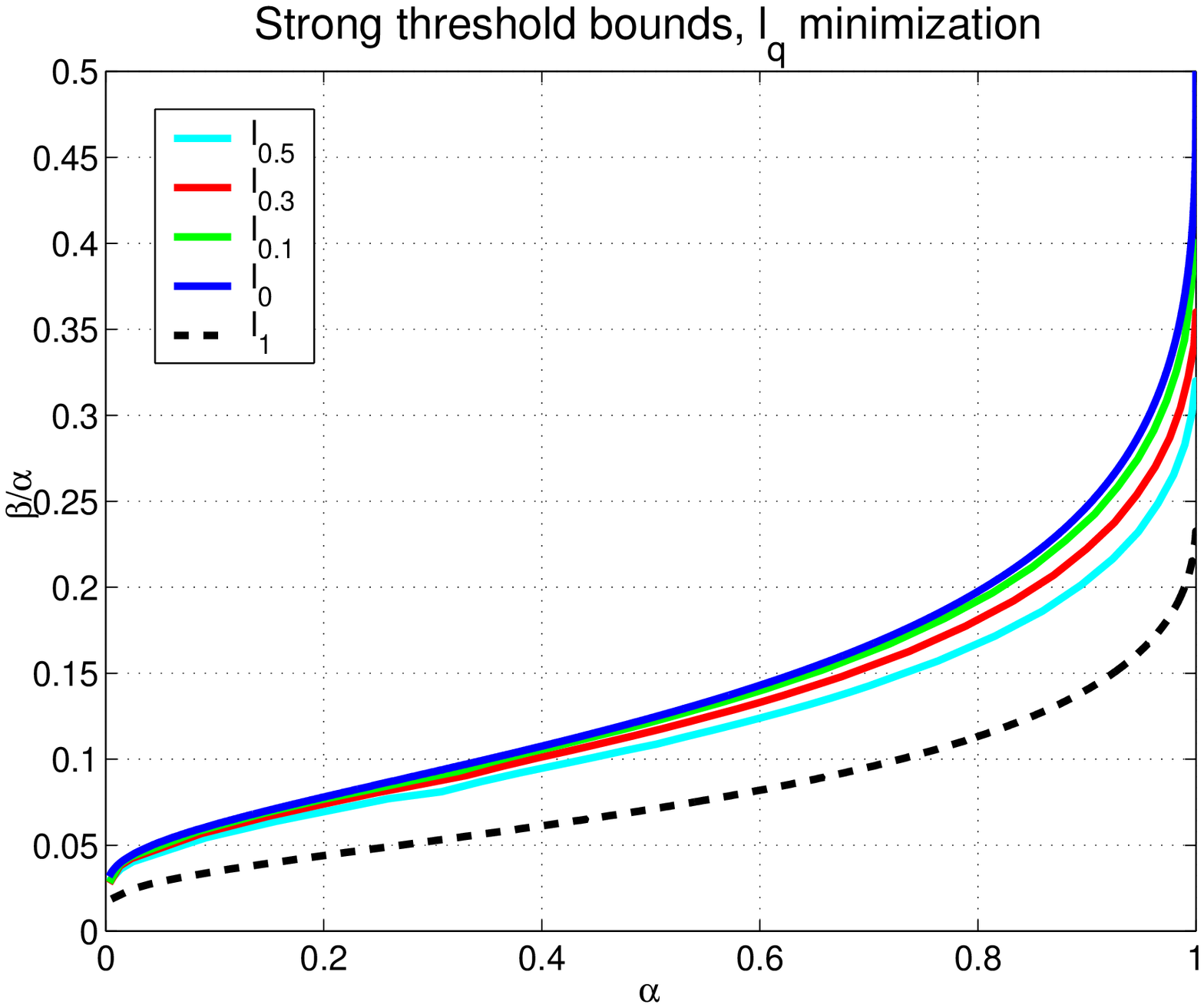,width=8cm,height=6.5cm}}
\end{minipage}
\begin{minipage}[b]{.5\linewidth}
\centering
\centerline{\epsfig{figure=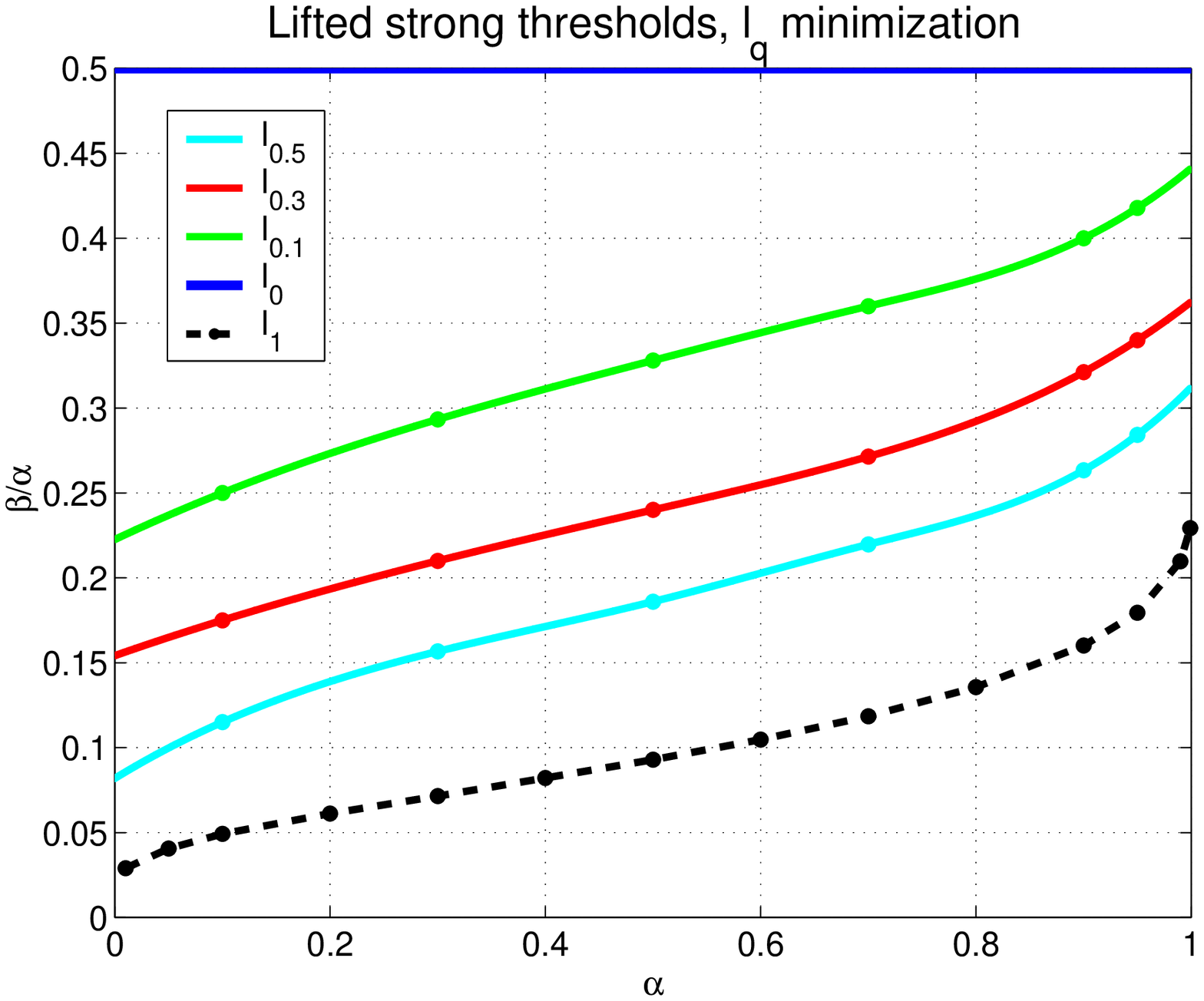,width=8cm,height=6.5cm}}
\end{minipage}
\caption{\emph{Strong} thresholds, $\ell_q$ optimization; a) left -- $c_3\rightarrow 0$; b) right -- optimized $c_3$}
\label{fig:str}
\end{figure}

As can be seen from Figure \ref{fig:str}, the results for selected values of $q$ are better than for $q=1$. Also the results improve on those presented in \cite{StojnicLqThrBnds10} and essentially obtained based on Corollary \ref{cor:corstrthrlq}, i.e. Theorem \ref{thm:thmstrthrlq} for $c_3^{(s)}\rightarrow 0$.

Also, we should emphasize that all the remarks related to numerical precisions/imprecisions we made when presenting results for the sectional thresholds in the previous section remain valid here as well. In fact, obtaining numerical results for the strong thresholds based on Theorem \ref{thm:thmstrthrlq} is even harder than obtaining the corresponding sectional ones using Theorem \ref{thm:thmsecthrlq} (essentially, one now has an extra optimization to do). So, one should again be careful when interpreting the presented results. They are again given more as an illustration so that the above theorem does not appear dry. It is on the other hand a very serious numerical analysis problem to actually obtain the numerical values for the thresholds based on the above theorem. We will investigate it in a greater detail elsewhere; here we only attempted to give a flavor as to what one can expect for these results to be.

Also, as mentioned earlier, all possible sub-optimal values that we obtained certainly don't jeopardize the rigorousness of the lower-bounding concept that we presented. However, the numerical integrations and possible finite precision errors when globally optimizing over $\w$ may contribute to curves being higher than they should. We however, firmly believe that this is not the case (or if it is that it is not to a drastic extent). As for how far away from the optimal thresholds the presented curves are, we do not know that. Conceptually however, the results presented in Theorem \ref{thm:thmstrthrlq} are probably not that far away from the optimal ones.

\subsection{Special case}
\label{sec:strthrspecial}
%%%%%%%%%%%%%%%%%%%%%%%%%%%%%%%%%%%%%%%%%%%%%%%%%%%%%%%%%%%%%%%%%

Similarly to what we did when we studied sectional thresholds in Section \ref{sec:secthr}, in this subsection we look at a couple of special cases for which the string thresholds can be computed more efficiently.

%%%%%%%%%%%%%%%%%%%%%%%%%%%%%%%%%%%%%%%%%%%%%%%%%%%%%%%%%%%%%%%%%
\subsubsection{$q\rightarrow 0$}
\label{sec:strthrspecialq0}
%%%%%%%%%%%%%%%%%%%%%%%%%%%%%%%%%%%%%%%%%%%%%%%%%%%%%%%%%%%%%%%%%

We will consider case $q\rightarrow 0$. As was the case when we studied sectional thresholds, there are many methods how the strong thresholds for $q\rightarrow 0$ case can be handled. Rather than obtaining the exact threshold results our goal is again to see what kind of performance would the methodology presented above give when $q\rightarrow 0$.

We of course again closely follow the methodology introduced above. As earlier, we will need a few modifications though. We start by introducing set $\Sstr^{(0)}$
\begin{equation}
\Sstr^{(0)}=\{\w^{(0)}\in S^{n-1}| \w_i^{(0)}=\b_i\w_i,1\leq i\leq n,\sum_{i=1}^{n}\b_i=2k;\sum_{i=1}^{n}\w_i^2=1\}.\label{eq:defSstr0}
\end{equation}
It is not that hard to see that when $q\rightarrow 0$ the above set can be used to characterize strong failure of $\ell_0$ optimization in a manner similar to the one set $\Sstr$ was used earlier to characterize strong failure of $\ell_q$ for a general $q$. Let $f(\w^{(0)})=\h^T\w^{(0)}$ and
we start with the following line of identities
\begin{multline}
\hspace{-.5in}\max_{\w^{(0)}\in\Sstr^{(0)}}f(\w^{(0)})=-\min_{\w^{(0)}\in\Sstr^{(0)}} -\h^T\w^{(0)}=-\min_{\w}\max_{\gamma_{str}\geq 0,\nu_{str}^{(0)}\geq 0}
-\sum_{i=1}^{n}\h_i\b_i\w_i
+\nu_{str}^{(0)}\sum_{i=1}^{n}\b_i
-\nu_{str}^{(0)}2k+\gamma_{str}\sum_{i=1}^{n}\w_i^2-\gamma_{str}\\
\leq-\max_{\gamma_{str}\geq 0,\nu_{str}^{(0)}\geq 0}\min_{\w}
-\sum_{i=1}^{n}\h_i\b_i\w_i
+\nu_{str}^{(0)}\sum_{i=1}^{n}\b_i
-\nu_{str}^{(0)}2k+\gamma_{str}\sum_{i=1}^{n}\w_i^2-\gamma_{str}\\
=\min_{\gamma_{str}\geq 0,\nu_{str}^{(0)}\geq 0}\max_{\w}
-\sum_{i=1}^{n}\h_i\b_i\w_i
-\nu_{str}^{(0)}\sum_{i=1}^{n}\b_i
+\nu_{str}^{(0)}2k-\gamma_{str}\sum_{i=1}^{n}\w_i^2+\gamma_{str}\\
=\min_{\gamma_{str}\geq 0,\nu_{str}^{(0)}\geq 0}
\sum_{i=1}^{n}\max\{\frac{\h_i^2}{4\gamma_{str}}-\nu_{str}^{(0)},0\}
+\nu_{str}^{(0)}2k+\gamma_{str}
=\min_{\gamma_{str}\geq 0,\nu_{str}^{(0)}\geq 0} f_2^{(0)}(\h,\nu_{str},\gamma_{str},\beta)+\gamma_{str},\label{eq:seceq1q0str}
%=\max_{\gamma_{str}\geq 0}(-\frac{\|\g\|_2^2}{4\gamma_{sph}}-\gamma_{sph})
\end{multline}
where
\begin{equation}
f_2^{(0)}(\h,\nu_{str},\gamma_{str},\beta)=\left (\sum_{i=1}^{n}\max\{\frac{\h_i^2}{4\gamma_{str}}-\nu_{str}^{(0)},0\}+\nu_{str}^{(0)}2k\right ).\label{eq:deff1q0str}
\end{equation}
Now one can write analogously to (\ref{eq:gamaiden1str})
\begin{equation}
I_{str}^{(0)}(c_3^{(s)},\beta)
\doteq \min_{\gamma_{str},\nu_{str}\geq 0}(\frac{\gamma_{str}}{\sqrt{n}}+\frac{1}{nc_3^{(s)}}\log(Ee^{c_3^{(s)}\sqrt{n}(f_1^{(0)}(\h,\nu_{str},\gamma_{str},\beta))})).\label{eq:gamaiden1secq0}
\end{equation}
After further introducing $\gamma_{str}=\gamma_{str}^{(s)}\sqrt{n}$, and $\nu_{str}^{(0)}=\nu_{str}^{(0,s)}\sqrt{n}^{-1}$ (where $\gamma_{str}^{(s)}$, and $\nu_{str}^{(0,s)}$ are independent of $n$) one can write analogously to (\ref{eq:gamaiden2sec})
\begin{multline}
I_{str}^{(0)}(c_3^{(s)},\beta)
\doteq\min_{\gamma_{str},\nu_{str}\geq 0}(\frac{\gamma_{str}}{\sqrt{n}}+\frac{1}{nc_3^{(s)}}\log(Ee^{c_3^{(s)}\sqrt{n}(f_1^{(0)}(\h,\nu_{str},\gamma_{str},\beta))})\\
=\min_{\gamma_{str}^{(s)},\nu_{str}^{(0,s)}\geq 0}(\gamma_{str}^{(s)}+2\beta\nu_{str}^{(0,s)}+\frac{1}{c_3^{(s)}}\log(Ee^{(c_3^{(s)}
\max\{\frac{\h_i^2}{\gamma_{str}^{(s)}}-\nu_{str}^{(0,s)},0\})}))
=\min_{\gamma_{str}^{(s)},\nu_{str}^{(0,s)}\geq 0}(\gamma_{str}^{(s)}+2\beta\nu_{str}^{(0,s)}+\frac{1}{c_3^{(s)}}\log(I_{sec}^{(0,1)})),\\\label{eq:gamaiden2secq0}
\end{multline}
where
\begin{equation}
I_{str}^{(0,1)}  =  Ee^{(c_3^{(s)}\max\{\frac{\h_i^2}{\gamma_{str}^{(s)}}-\nu_{str}^{(0,s)},0\})}.\label{eq:defI1I2strq0}
\end{equation}
One can then write analogously to (\ref{eq:seccondthmstr})
\begin{equation}
\min_{c_3^{(s)}}(-\frac{c_3^{(s)}}{2}+I_{str}^{(0)}(c_3^{(s)},\beta)+I_{sph}(c_3^{(s)},\alpha))<0.\label{eq:seccondthmstrq0}
\end{equation}
Setting $b=\frac{c_3^{(s)}}{4\gamma_{str}}$, $\nu_{str}^{(0,s,\gamma)}=4\gamma_{str}\nu_{str}^{(0,s)}$, and solving the integrals one from (\ref{eq:seccondthmstrq0}) has the following condition for $\beta$ and $\alpha$
\begin{multline}
\frac{2b\nu_{str}^{(0,s,\gamma)}\beta}{c_3^{(s)}}+c_3^{(s)}\frac{1-2b}{4b}\\+\frac{1}{c_3^{(s)}}\log \left (\frac{e^{-b\nu_{str}^{(0,s,\gamma)}}}{\sqrt{1-2b}}\mbox{erfc}\left (\sqrt{\frac{1-2b}{2}\nu_{str}^{(0,s,\gamma)}}\right )
+\mbox{erf}\left (\sqrt{\frac{\nu_{str}^{(0,s,\gamma)}}{2}}\right )\right )+I_{sph}(c_3^{(s)},\alpha)<0.\label{q0condIsph1str}
\end{multline}
Now, assuming $c_3^{(s)}$ is large one has
\begin{equation}
I_{sph}(c_3^{(s)},\alpha)\approx -\frac{\alpha}{2c_3^{(s)}}-\frac{\alpha}{2c_3^{(s)}}\log(1+\frac{(c_3^{(s)})^2}{\alpha}).\label{q0condIsph2str}
\end{equation}
Setting
\begin{eqnarray}
1-2b & = & \frac{\alpha}{(c_3^{(s)})^2}\nonumber \\
\nu_{str}^{(0,s,\gamma)} & = & \log(\frac{(c_3^{(s)})^2}{\alpha}),\label{eq:q0condstr}
\end{eqnarray}
one from (\ref{q0condIsph1str}) and (\ref{q0condIsph2str}) has
\begin{multline}
\frac{2b\nu_{str}^{(0,s,\gamma)}\beta}{c_3^{(s)}}+c_3^{(s)}\frac{1-2b}{4b}\\+\frac{1}{c_3^{(s)}}\log \left (\frac{e^{-b\nu_{str}^{(0,s,\gamma)}}}{\sqrt{1-2b}}\mbox{erfc}\left (\sqrt{\frac{1-2b}{2}\nu_{str}^{(0,s,\gamma)}}\right )
+\mbox{erf}\left (\sqrt{\frac{\nu_{str}^{(0,s,\gamma)}}{2}}\right )\right )+I_{sph}(c_3^{(s)},\alpha)
= O\left (\frac{(2\beta-\alpha)\log(c_3^{(s)})}{c_3^{(s)}}\right ).\\\label{q0condIsph3str}
\end{multline}
Then from (\ref{q0condIsph3str}) one has that as long as $\beta_{str}^{(0)}<\frac{\alpha}{2}-\epsilon_{\ell_0}$, where $\epsilon_{\ell_0}$ is a small positive constant (adjusted with respect to $c_{3}^{(s)}$) (\ref{eq:seccondthmstrq0}) holds which is, as stated above, an analogue to the condition given in (\ref{eq:seccondthmstr}) in Theorem \ref{thm:thmstrthrlq}. This essentially means that when $q\rightarrow 0$ one has that the threshold curve approaches the best possible curve $\frac{\beta}{\alpha}=\frac{1}{2}$. This is the same conclusion we achieved in Section \ref{sec:secthr} when we studied sectional thresholds. Of course, as stated a couple of times earlier, if our goal was to show what is the best curve one can achieve when $q\rightarrow 0$ we would not need all of the machinery that we just used. However, the idea was different. We essentially wanted to show what are the limits of the methodology that we introduced in this paper. It turns out that when $q\rightarrow 0$ our methodology is good enough to recover the best possible threshold curve. It is though not as likely that this is the case for any other $q$.

\subsubsection{$q=\frac{1}{2}$}
\label{sec:strthrspecialq05}
%%%%%%%%%%%%%%%%%%%%%%%%%%%%%%%%%%%%%%%%%%%%%%%%%%%%%%%%%%%%%%%%%

As was the case when we studied sectional thresholds, one can also make substantial simplifications when $q=\frac{1}{2}$. However, the remaining integrals are still quite involved and skip presenting this easy but tedious exercise.

%%%%%%%%%%%%%%%%%%%%%%%%%%%%%%%%%%%%%%%%%%%%%%%%%%%%%%%%%%%%%%%%%
\section{$\ell_q$-minimization weak threshold}
\label{sec:weakthr}
%%%%%%%%%%%%%%%%%%%%%%%%%%%%%%%%%%%%%%%%%%%%%%%%%%%%%%%%%%%%%%%%%

In this section we at the weak thresholds of $\ell_q$ minimization. As earlier, we will slit the presentation into two parts; the first one will introduce a few preliminary results and the second one will contain the main arguments.

%%%%%%%%%%%%%%%%%%%%%%%%%%%%%%%%%%%%%%%%%%%%%%%%%%%%%%%%%%%%%%%%%
\subsection{Weak threshold preliminaries}
\label{sec:weakthrprelim}
%%%%%%%%%%%%%%%%%%%%%%%%%%%%%%%%%%%%%%%%%%%%%%%%%%%%%%%%%%%%%%%%%

Below we will present a way to quantify behavior of $\beta_{weak}^{(q)}$. As usual, we rely on some of the mechanisms presented in \cite{StojnicCSetam09}, some of those presented in Section \ref{sec:secthr}, and some of those presented in \cite{StojnicLqThrBnds10}. We start by introducing a nice way of characterizing weak success/failure of (\ref{eq:lq}).

\begin{theorem}(A given fixed $\x$ \cite{StojnicLqThrBnds10})
Assume that an $m\times n$ matrix $A$ is given. Let $\tilde{\x}$ be a $k$-sparse vector and let $\tilde{\x}_1=\tilde{\x}_2=\dots=\tilde{\x}_{n-k}=0$. Further, assume that $\y=A\tilde{\x}$ and that $\w$ is
an $n\times 1$ vector. If
\begin{equation}
(\forall \w\in \textbf{R}^n | A\w=0) \quad  \sum_{i=1}^{n-k}|\w_i|^q+\sum_{i=n-k+1}^n|\tilde{\x}_i+\w_i|^q>\sum_{i=n-k+1}^{n}|\tilde{\x}_{i}|^q
\label{eq:thmeqgenweak1}
\end{equation}
then the solution of (\ref{eq:lq}) obtained for pair $(\y,A)$ is $\tilde{\x}$.
%Moreover, if
%\begin{equation}
%(\exists \w\in \textbf{R}^n | A\w=0) \quad  \sum_{i=n-k+1}^n |\w_i|>\sum_{i=1}^{n-k}|\w_{i}|
%\label{eq:thmeqgenweak2}
%\end{equation}
%then there will be a $k$-sparse $\x$ that satisfies (\ref{eq:system}) and is not the solution of (\ref{eq:l1}).
\label{thm:thmgenweak}
\end{theorem}
\begin{proof}
The proof is of course very simple and for completeness is included in \cite{StojnicLqThrBnds10}.
\end{proof}

We then, following the methodology of the previous section (and ultimately of \cite{StojnicCSetam09,StojnicLqThrBnds10}),
start by defining a set $\Sweak$
\begin{equation}
\Sweak(\tilde{\x})=\{\w\in S^{n-1}| \quad \sum_{i=n-k+1}^n |\tilde{\x}_i|^q\geq \sum_{i=1}^{n-k}|\w_i|^q+\sum_{i=n-k+1}^n|\tilde{\x}_i+\w_i|^q\},\label{eq:defSweak}
\end{equation}
where $S^{n-1}$ is the unit sphere in $R^n$. The methodology of the previous section (and ultimately the one of \cite{StojnicCSetam09,StojnicLqThrBnds10}) then proceeds by considering the following optimization problem
\begin{equation}
\xi_{weak}(\tilde{\x})=\min_{\w\in\Sweak(\tilde{\x})}\|A\w\|_2,\label{eq:negham1weak}
\end{equation}
where $q=1$ in the definition of $\Sweak$ (the same will remain true for any $0\leq q\leq 1$). One can then argue as in the previous sections: if $\xi_{weak}$ is positive with overwhelming probability for certain combination of $k$, $m$, and $n$ then for $\alpha=\frac{m}{n}$ one has a lower bound $\beta_{weak}=\frac{k}{n}$ on the true value of the weak threshold with overwhelming probability. Following \cite{StojnicCSetam09} one has that $\xi_{weak}$ concentrates, which essentially means that if we can show that $\min_{\tilde{\x}}(E(\xi_{weak}(\tilde{\x})))>0$ for certain $k$, $m$, and $n$ we can then obtain a lower bound on the weak threshold. In fact, this is precisely what was done in \cite{StojnicCSetam09}. Moreover, as shown in \cite{StojnicUpper10}, the results obtained in \cite{StojnicCSetam09} are actually exact.
The main reason of course was ability to determine $E\xi_{weak}$ exactly.

When it comes to the weak thresholds for a general $q$ we in \cite{StojnicLqThrBnds10} adopted the strategy similar to the one employed in \cite{StojnicCSetam09}. However, the results we obtained through such a consideration were not exact. The main reason was an inability to determine $E\xi_{weak}$ exactly for a general $q<1$. We were then left with the lower bounds which turned out to be loose. In this section we will use some of the ideas from the previous section (and essentially those from \cite{StojnicMoreSophHopBnds10,StojnicLiftStrSec13}) to provide a substantial conceptual improvements on bounds given in \cite{StojnicLqThrBnds10}. A limited numerical exploration also indicates that they are likely in turn to reflect in a practical improvement of the weak thresholds as well.

We start by emulating what was done in the previous sections, i.e. by presenting a way to create a lower-bound on the optimal value of (\ref{eq:negham1weak}).

%%%%%%%%%%%%%%%%%%%%%%%%%%%%%%%%%%%%%%%%%%%%%%%%%%%%%%%%%%%%%%%%%
\subsection{Lower-bounding $\xi_{weak}$}
\label{sec:lbxiweak}
%%%%%%%%%%%%%%%%%%%%%%%%%%%%%%%%%%%%%%%%%%%%%%%%%%%%%%%%%%%%%%%%%

In this section we will look at the problem from (\ref{eq:negham1str}). We recall that as earlier, we will consider a statistical scenario and assume that the elements of $A$ are i.i.d. standard normal random variables. Such a scenario was considered in \cite{StojnicLiftStrSec13} as well and the following was done.
First we reformulated the problem in (\ref{eq:negham1weak}) in the following way
\begin{equation}
\xi_{weak}=\min_{\w\in\Sweak}\max_{\|\y\|_2=1}\y^TA\w.\label{eq:sqrtnegham2weak}
\end{equation}
Then using results of \cite{StojnicHopBnds10} we established a lemma very similar to the following one:
\begin{lemma}
Let $A$ be an $m\times n$ matrix with i.i.d. standard normal components. Let $\Sweak(\tilde{\x})$ be a collection of sets defined in (\ref{eq:defSweak}). Let $\g$ and $\h$ be $n\times 1$ and $m\times 1$ vectors, respectively, with i.i.d. standard normal components. Also, let $g$ be a standard normal random variable and let $c_3$ be a positive constant. Then
\begin{equation}
\max_{\tilde{\x}}E(\max_{\w\in\Sweak}\min_{\|\y\|_2=1}e^{-c_3(\y^T A\w + g)})\leq \max_{\tilde{\x}}E(\max_{\w\in\Ssec}\min_{\|\y\|_2=1}e^{-c_3(\g^T\y+\h^T\w)}).\label{eq:negexplemmaweak}
\end{equation}\label{lemma:negexplemmaweak}
\end{lemma}
\begin{proof}
As mentioned in the previous sections (as well as in \cite{StojnicLiftStrSec13} and earlier in \cite{StojnicHopBnds10}), the proof is a standard/direct application of a theorem from \cite{Gordon85}. We omit the details.
\end{proof}

Following what was done after Lemma 3 in \cite{StojnicHopBnds10} one arrives at the following analogue of \cite{StojnicHopBnds10}'s equation $(57)$:
\begin{equation}
E(\min_{\w\in\Sweak}\|A\w\|_2)\geq
\frac{c_3}{2}-\frac{1}{c_3}\log(E(\max_{\w\in\Sweak}(e^{-c_3\h^T\w})))
-\frac{1}{c_3}\log(E(\min_{\|\y\|_2=1}(e^{-c_3\g^T\y}))).\label{eq:chneg8weak}
\end{equation}
Let $c_3=c_3^{(s)}\sqrt{n}$ where $c_3^{(s)}$ is a constant independent of $n$. Then (\ref{eq:chneg8weak}) becomes
\begin{eqnarray}
\hspace{-.5in}\frac{E(\min_{\w\in\Sweak}\|A\w\|_2)}{\sqrt{n}}
& \geq &
\frac{c_3^{(s)}}{2}-\frac{1}{nc_3^{(s)}}\log(E(\max_{\w\in\Sweak}(e^{-c_3^{(s)}\h^T\w})))
-\frac{1}{nc_3^{(s)}}\log(E(\min_{\|\y\|_2=1}(e^{-c_3^{(s)}\sqrt{n}\g^T\y})))\nonumber \\
& = &-(-\frac{c_3^{(s)}}{2}+I_{weak}(c_3^{(s)},\beta)+I_{sph}(c_3^{(s)},\alpha)),\label{eq:chneg9weak}
\end{eqnarray}
where
\begin{eqnarray}
I_{weak}(c_3^{(s)},\beta) & = & \frac{1}{nc_3^{(s)}}\log(E(\max_{\w\in\Sweak}(e^{-c_3^{(s)}\h^T\w})))\nonumber \\
I_{sph}(c_3^{(s)},\alpha) & = & \frac{1}{nc_3^{(s)}}\log(E(\min_{\|\y\|_2=1}(e^{-c_3^{(s)}\sqrt{n}\g^T\y}))).\label{eq:defIsstr}
\end{eqnarray}

As in previous section, the above bound is effectively correct for any positive constant $c_3^{(s)}$. To make it operational one needs to estimate $I_{weak}(c_3^{(s)},\beta)$ and $I_{sph}(c_3^{(s)},\alpha)$. Of course, $I_{sph}(c_3^{(s)},\alpha)$ has already been characterized in (\ref{eq:gamaiden3}) and (\ref{eq:Isph}). That basically means that the only thing that is left to characterize is $I_{weak}(c_3^{(s)},\beta)$. To facilitate the exposition we will, as earlier, focus only on the large $n$ scenario. Let $f(\w)=-\h^T\w$. Following \cite{StojnicLqThrBnds10} one can arrive at
\begin{equation}
\max_{\w\in\Sweak}f(\w)=-\min_{\w\in\Sweak} -\h^T\w
\leq\min_{\gamma_{weak}\geq 0,\nu_{weak}\geq 0} f_3(q,\h,\nu_{weak},\gamma_{weak},\beta)+\gamma_{weak},\label{eq:seceq1weak}
%=\max_{\gamma_{sec}\geq 0}(-\frac{\|\g\|_2^2}{4\gamma_{sph}}-\gamma_{sph})
\end{equation}
where
\begin{multline}
f_3(q,\h,\nu_{weak},\gamma_{weak},\beta)=\max_{\w} (\sum_{i=n-k+1}^{n}(\h_i\w_i-\nu_{weak}|\tilde{\x}_i+\w_i|^q+\nu_{weak}|\tilde{\x}_i|^q-\gamma_{weak}\w_i^2)\\
+\sum_{i=1}^{n-k}(\h_i|\w_i|-\nu_{weak}|\w_i|^q-\gamma_{weak}\w_i^2)).\label{eq:deff1weak}
\end{multline}
Then
\begin{multline}
I_{weak}(c_3^{(s)},\beta)  =  \frac{1}{nc_3^{(s)}}\log(E(\max_{\w\in\Sweak}(e^{-c_3^{(s)}\h^T\w}))) = \frac{1}{nc_3^{(s)}}\log(E(\max_{\w\in\Sweak}(e^{c_3^{(s)}f(\w))})))\\
=\frac{1}{nc_3^{(s)}}\log(Ee^{c_3^{(s)}\sqrt{n}\min_{\gamma_{weak},\nu_{weak}\geq 0}(f_3(\h,\nu_{weak},\gamma_{weak},\beta)+\gamma_{weak})})\\
\doteq \frac{1}{nc_3^{(s)}}\min_{\gamma_{weak},\nu_{weak}\geq 0}\log(Ee^{c_3^{(s)}\sqrt{n}(f_3(q,\h,\nu_{weak},\gamma_{weak},\beta)+\gamma_{weak})})\\
=\min_{\gamma_{weak},\nu_{weak}\geq 0}(\frac{\gamma_{weak}}{\sqrt{n}}+\frac{1}{nc_3^{(s)}}\log(Ee^{c_3^{(s)}\sqrt{n}(f_3(q,\h,\nu_{weak},\gamma_{weak},\beta))})),\label{eq:gamaiden1weak}
\end{multline}
where, as earlier, $\doteq$ stands for equality when $n\rightarrow \infty$. Now if one sets $\w_{i}=\frac{\w_{i}^{(s)}}{\sqrt{n}}$, $\gamma_{weak}=\gamma_{weak}^{(s)}\sqrt{n}$, and $\nu_{weak}=\nu_{weak}^{(s)}\sqrt{n}^{q-1}$ (where $\w_{i}^{(s)}$, $\gamma_{weak}^{(s)}$, and $\nu_{weak}^{(s)}$ are independent of $n$) then (\ref{eq:gamaiden1weak}) gives
\begin{multline}
I_{weak}(c_3^{(s)},\beta)
=\min_{\gamma_{weak},\nu_{weak}\geq 0}(\frac{\gamma_{weak}}{\sqrt{n}}+\frac{1}{nc_3^{(s)}}\log(Ee^{c_3^{(s)}\sqrt{n}(f_3(q,\h,\nu_{weak},\gamma_{weak},\beta))})\\
\hspace{-.5in}=\min_{\gamma_{weak}^{(s)},\nu_{weak}^{(s)}\geq 0}(\gamma_{weak}^{(s)}+\frac{\beta}{c_3^{(s)}}\log\left (Ee^{\left (c_3^{(s)}
\max_{\w_i^{(s)}}(\h_i\w_i^{(s)}-\nu_{weak}^{(s)}|\tilde{\x}_i+\w_i^{(s)}|^q+\nu_{weak}^{(s)}|\tilde{\x}_i|^q-\gamma_{weak}^{(s)}(\w_i^{(s)})^2)\right )}\right )\\
+\frac{1-\beta}{c_3^{(s)}}\log\left (Ee^{\left (c_3^{(s)}\max_{\w_j^{(s)}}(|\h_j||\w_j^{(s)}|-\nu_{weak}^{(s)}|\w_j^{(s)}|^q-\gamma_{weak}^{(s)}(\w_j^{(s)})^2)
\right )}\right ))\\
\\=\min_{\gamma_{weak}^{(s)},\nu_{weak}^{(s)}\geq 0}(\gamma_{weak}^{(s)}+\frac{\beta}{c_3^{(s)}}\log(I_{weak}^{(1)})
+\frac{1-\beta}{c_3^{(s)}}\log(I_{weak}^{(2)})),\\\label{eq:gamaiden2weak}
\end{multline}
where
\begin{eqnarray}
I_{weak}^{(1)} & = & Ee^{\left (c_3^{(s)}
\max_{\w_i^{(s)}}(\h_i\w_i^{(s)}-\nu_{weak}^{(s)}|\tilde{\x}_i+\w_i^{(s)}|^q+\nu_{weak}^{(s)}|\tilde{\x}_i|^q-\gamma_{weak}^{(s)}(\w_i^{(s)})^2)\right )}\nonumber \\
I_{weak}^{(2)} & = & Ee^{\left (c_3^{(s)}\max_{\w_j^{(s)}}(|\h_j||\w_j^{(s)}|-\nu_{weak}^{(s)}|\w_j^{(s)}|^q-\gamma_{weak}^{(s)}(\w_j^{(s)})^2)
\right )}.\label{eq:defI1I2weak}
\end{eqnarray}

We summarize the above results related to the weak threshold ($\beta_{weak}^{(q)}$) in the following theorem.

\begin{theorem}(Weak threshold - lifted lower bound)
Let $A$ be an $m\times n$ measurement matrix in (\ref{eq:system})
with i.i.d. standard normal components. Let $\tilde{\x}\in R^n$ be a $k$-sparse vector for which $\tilde{\x}_1=0,\tilde{\x}_2=0,,\dots,\tilde{\x}_{n-k}=0$ and let  $\y=A\tilde{\x}$. Let $k,m,n$ be large
and let $\alpha=\frac{m}{n}$ and $\betaweak^{(q)}=\frac{k}{n}$ be constants
independent of $m$ and $n$. Let $c_3^{(s)}$ be a positive constant and set
\begin{equation}
\widehat{\gamma_{sph}^{(s)}}=\frac{2c_3^{(s)}-\sqrt{4(c_3^{(s)})^2+16\alpha}}{8},\label{eq:gamaiden3thmweak}
\end{equation}
and
\begin{equation}
I_{sph}(c_3^{(s)},\alpha)=
\left ( \widehat{\gamma_{sph}^{(s)}}-\frac{\alpha}{2c_3^{(s)}}\log(1-\frac{c_3^{(s)}}{2\widehat{\gamma_{sph}^{(s)}}}\right ).\label{eq:Isphthmweak}
\end{equation}
Further let
\begin{eqnarray}
I_{weak}^{(1)} & = & Ee^{\left (c_3^{(s)}
\max_{\w_i^{(s)}}(\h_i\w_i^{(s)}-\nu_{weak}^{(s)}|\tilde{\x}_i+\w_i^{(s)}|^q+\nu_{weak}^{(s)}|\tilde{\x}_i|^q-\gamma_{weak}^{(s)}(\w_i^{(s)})^2)\right )}\nonumber \\
I_{weak}^{(2)} & = & Ee^{\left (c_3^{(s)}\max_{\w_j^{(s)}}(|\h_j||\w_j^{(s)}|-\nu_{weak}^{(s)}|\w_j^{(s)}|^q-\gamma_{weak}^{(s)}(\w_j^{(s)})^2)
\right )},\label{eq:defI1I2secthmweak}
\end{eqnarray}
and
\begin{equation}
I_{weak}(c_3^{(s)},\betaweak^{(q)})=\min_{\gamma_{weak}^{(s)},\nu_{weak}^{(s)}\geq 0}(\gamma_{weak}^{(s)}+\frac{\betaweak^{(q)}}{c_3^{(s)}}\log(I_{weak}^{(1)})+\frac{1-\betaweak^{(q)}}{c_3^{(s)}}\log(I_{weak}^{(2)})).\label{eq:seccondthmweak}
\end{equation}
If $\alpha$ and $\betaweak^{(q)}$ are such that
\begin{equation}
\max_{\tilde{\x}_{i,i>n-k}}\min_{c_3^{(s)}}(-\frac{c_3^{(s)}}{2}+I_{weak}(c_3^{(s)},\betaweak^{(q)})+I_{sph}(c_3^{(s)},\alpha))<0,\label{eq:seccondthmweak}
\end{equation}
then with overwhelming probability the solution of (\ref{eq:lq}) for pair $(\y,A)$ is $\tilde{\x}$.\label{thm:thmweakthrlq}
\end{theorem}
\begin{proof}
Follows from the above discussion.
\end{proof}

One also has immediately the following corollary.

\begin{corollary}(Weak threshold - lower bound \cite{StojnicLqThrBnds10})
Let $A$ be an $m\times n$ measurement matrix in (\ref{eq:system})
with i.i.d. standard normal components. Let $\tilde{\x}\in R^n$ be a $k$-sparse vector for which $\tilde{\x}_1=0,\tilde{\x}_2=0,,\dots,\tilde{\x}_{n-k}=0$ and let  $\y=A\tilde{\x}$. Let $k,m,n$ be large
and let $\alpha=\frac{m}{n}$ and $\betaweak^{(q)}=\frac{k}{n}$ be constants
independent of $m$ and $n$. Let
\begin{equation}
I_{sph}(\alpha)=
-\sqrt{\alpha}.\label{eq:Isphcorweak}
\end{equation}
Further let
\begin{eqnarray}
I_{weak}^{(1)} & = & E
\max_{\w_i^{(s)}}(\h_i\w_i^{(s)}-\nu_{weak}^{(s)}|\tilde{\x}_i+\w_i^{(s)}|^q+\nu_{weak}^{(s)}|\tilde{\x}_i|^q-\gamma_{weak}^{(s)}(\w_i^{(s)})^2)\nonumber \\
I_{weak}^{(2)} & = & E\max_{\w_j^{(s)}}(|\h_j||\w_j^{(s)}|-\nu_{weak}^{(s)}|\w_j^{(s)}|^q-\gamma_{weak}^{(s)}(\w_j^{(s)})^2),\label{eq:defI1I2weakcor}
\end{eqnarray}
and
\begin{equation}
I_{weak}(\betaweak^{(q)})=\min_{\gamma_{weak}^{(s)},\nu_{weak}^{(s)}\geq 0}(\gamma_{weak}^{(s)}+\betaweak^{(q)}I_{weak}^{(1)}+(1-\betaweak^{(q)})I_{weak}^{(2)}).\label{eq:seccondcorweak}
\end{equation}
If $\alpha$ and $\betaweak^{(q)}$ are such that
\begin{equation}
\max_{\tilde{\x}_{i,i>n-k}}(I_{weak}(\betaweak^{(q)})+I_{sph}(\alpha))<0,\label{eq:seccondcorweak}
\end{equation}
then with overwhelming probability the solution of (\ref{eq:lq}) for pair $(\y,A)$ is $\tilde{\x}$.\label{cor:corweakthrlq}
\end{corollary}
\begin{proof}
Follows from the above theorem by taking $c_3^{(s)}\rightarrow 0$.
\end{proof}

The results for the weak threshold obtained from the above theorem
are presented in Figure \ref{fig:weak}. To be a bit more specific, we selected four different values of $q$, namely $q\in\{0,0.1,0.3,0.5\}$ in addition to standard $q=1$ case already discussed in \cite{StojnicCSetam09}. Also, we present in Figure \ref{fig:weak} the results one can get from Theorem \ref{thm:thmweakthrlq} when $c_3^{(s)}\rightarrow 0$ (i.e. from Corollary \ref{cor:corweakthrlq}, see e.g. \cite{StojnicLqThrBnds10}).
\begin{figure}[htb]
\begin{minipage}[b]{.5\linewidth}
\centering
\centerline{\epsfig{figure=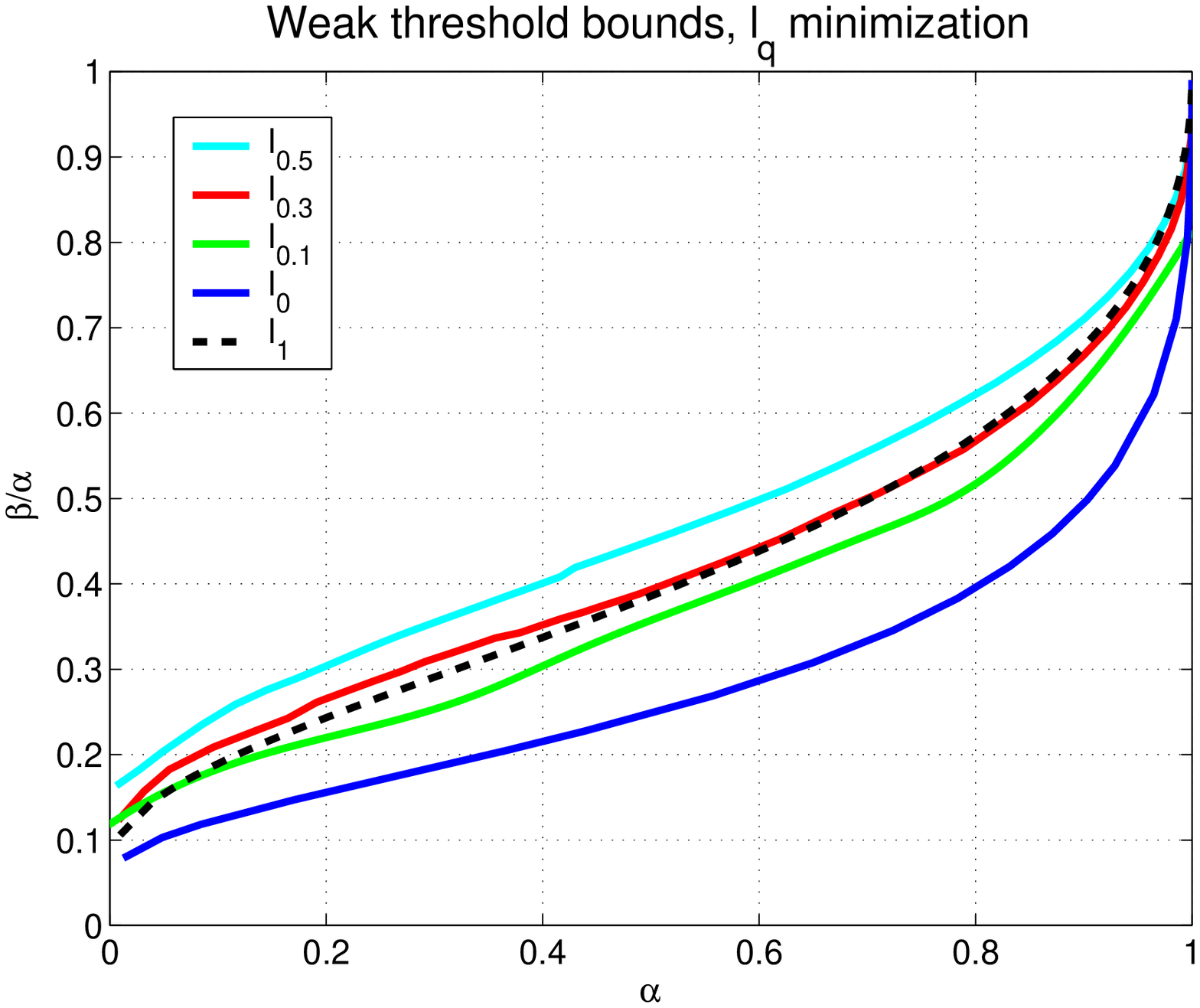,width=8cm,height=6.5cm}}
\end{minipage}
\begin{minipage}[b]{.5\linewidth}
\centering
\centerline{\epsfig{figure=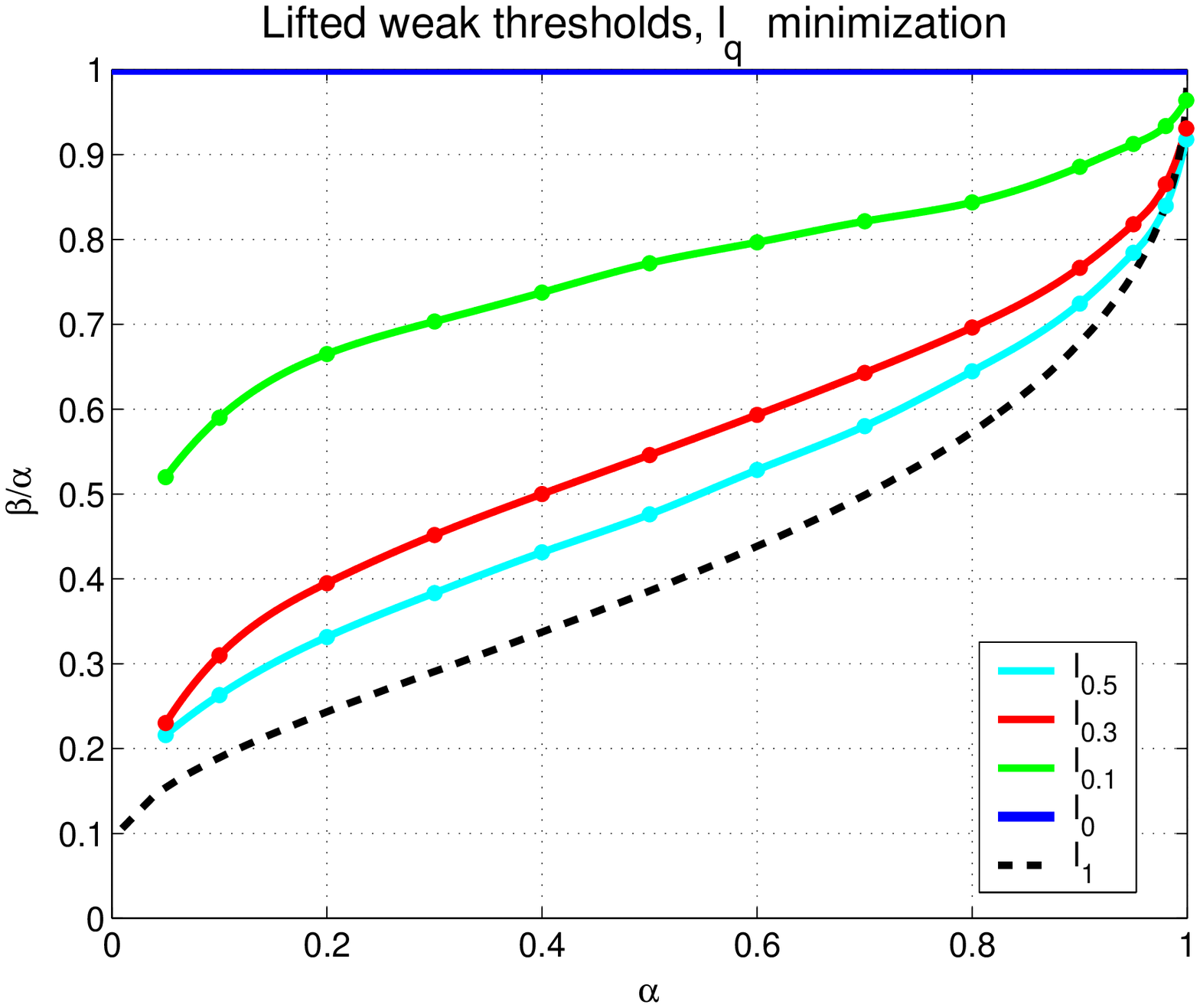,width=8cm,height=6.5cm}}
\end{minipage}
\caption{\emph{Weak} thresholds, $\ell_q$ optimization; a) left -- $c_3\rightarrow 0$; b) right -- optimized $c_3$}
\label{fig:weak}
\end{figure}

As can be seen from Figure \ref{fig:weak}, the results for selected values of $q$ are better than for $q=1$. Also the results improve on those presented in \cite{StojnicLqThrBnds10} and essentially obtained based on Corollary \ref{cor:corweakthrlq}, i.e. Theorem \ref{thm:thmweakthrlq} for $c_3^{(s)}\rightarrow 0$.

Also, we should again recall that all of presented results were obtained after numerical computations. These are on occasion even more involved than those presented in Section \ref{sec:secthr} and could be imprecise. In that light we would again suggest that one should take the results presented in Figure \ref{fig:sec} more as an illustration rather than as an exact plot of the achievable thresholds (this is especially true for curve $q=0.1$ since the smaller values of $q$ cause more numerical problems; in fact one can easily observe a slightly jittery shape of $q=0.1$ curves). Obtaining the presented results included several numerical optimizations which were all (except maximization over $\w$ and $\tilde{\x}$) done on a local optimum level. We do not know how (if in any way) solving them on a global optimum level would affect the location of the plotted curves. Also, additionally, numerical integrations were done on a finite precision level which could have potentially harmed the final results as well. Still, as earlier, we believe that the methodology can not achieve substantially more than what we presented in Figure \ref{fig:sec} (and hopefully is not severely degraded with numerical integrations and maximization over $\w$ and $\tilde{\x}$). Of course, we do reemphasize again that the results presented in Theorem \ref{thm:thmweakthrlq} are completely rigorous,
it is just that some of the numerical work that we performed could have been a bit imprecise.

%%%%%%%%%%%%%%%%%%%%%%%%%%%%%%%%%%%%%%%%%%%%%%%%%%%%%%%%%%%%%%%%%
\subsection{Special cases}
\label{sec:weakthrspecial}
%%%%%%%%%%%%%%%%%%%%%%%%%%%%%%%%%%%%%%%%%%%%%%%%%%%%%%%%%%%%%%%%%

One can again create a substantial simplification of results given in Theorem \ref{thm:thmweakthrlq} for certain values of $q$. For example, for $q=0$ or $q=1/2$ one can follow the strategy of previous sections and simplify some of the computations. However, such results (while simpler than those from Theorem \ref{thm:thmweakthrlq}) are still not very simple and we skip presenting them. We do mention, that this is in particular so since one also has to optimize over $\tilde{\x}$. We did however include the ideal plot for case $q=0$ in Figure \ref{fig:weak}.

\section{Conclusion}
\label{sec:conc}
%%%%%%%%%%%%%%%%%%%%%%%%%%%%%%%%%%%%%%%%%%%%%%%%%%%%%%%%%%%%%%%%%%%%%%%%%%%%%%%%

In this paper we looked at classical under-determined linear systems with sparse solutions. We analyzed a particular optimization technique called $\ell_q$ optimization. While its a convex counterpart $\ell_1$ technique is known to work well often it is a much harder task to determine if
$\ell_q$ exhibits a similar or better behavior; and especially if it exhibits a better behavior how much better quantitatively it is.

In our recent work \cite{StojnicLqThrBnds10} we made some sort of progress in this direction. Namely, in \cite{StojnicLqThrBnds10}, we showed that in many cases the $\ell_q$ would provide stronger guarantees than $\ell_1$ and in many other ones we provided bounds that are better than the ones we could provide for $\ell_1$. Of course, having better bounds does not guarantee that the performance is better as well but in our view it served as a solid indication that overall, $\ell_q,q<1$, should work better than $\ell_1$.

In this paper we went a few steps further and created a powerful mechanism to lift the threshold bounds we provided in \cite{StojnicLqThrBnds10}. While the results are theoretically rigorous and certainly provide a substantial conceptual progress, their practical usefulness is predicated on numerically solving a collection of optimization problems. We left such a detailed study for a forthcoming paper and here provided a limited set of numerical results we obtained. According to the results we provided one has a substantial improvement on the threshold bounds obtained in \cite{StojnicLqThrBnds10}. Moreover, one of the main issues that hindered a complete success of the technique used in \cite{StojnicLqThrBnds10} was a bit surprising non-monotonic change in thresholds behavior with respect to the value of $q$. Namely, in \cite{StojnicLqThrBnds10}, we obtained bounds that were improving as $q$ was going down (a fact expected based on tightening of the sparsity relaxation). However, such an improving was happening only until $q$ was reaching towards a certain limit. As $q$ was decreasing beyond such a limit the bounds started going down and eventually in the most trivial case $q=0$ they even ended up being worse than the ones we obtained for $q=1$. Based on our limited numerical results, the mechanisms we provided in this paper at the very least do not seem to suffer from this phenomenon. In other words, the numerical results we provided (if correct) indicate that as $q$ goes down all the thresholds considered in this paper indeed go up.

Another interesting point is of course from a purely theoretical side. That essentially means, leaving aside for a moment all the required numerical work and its precision, can one say what the ultimate capabilities of the theoretical results we provided in this paper are. This is actually fairly hard to assess even if we were able to solve all numerical problems with a full precision. While we have a solid belief that when $q=1$ a similar set of results obtained in \cite{StojnicLiftStrSec13} is fairly close to the optimal one, here it is not as clear. We do believe that the theoretical results we provided here are also close to the optimal ones but probably not as close as the ones given in \cite{StojnicLqThrBnds10} are to their corresponding optimal ones. Of course, to get a better insight how far off they could be one would have to implement further nested upgrades along the lines of what was discussed in \cite{StojnicLiftStrSec13}. That makes the numerical work infinitely many times more cumbersome and while we have done it to a degree for problems considered in \cite{StojnicLiftStrSec13} for those considered here we have not. As mentioned in \cite{StojnicLiftStrSec13}, designing such an upgrade is practically relatively easy. However, the number of optimizing variables grows fast as well and we did not find it easy to numerically handle even the number of variables that we have had here.

Of course, as was the case in \cite{StojnicLqThrBnds10}, much more can be done, including generalizations of the presented concepts to many other variants of these problems. The examples include various different unknown vector structures (a priori known to be positive vectors, block-sparse, binary/box constrained vectors etc.), various noisy versions (approximately sparse vectors, noisy measurements $\y$), low rank matrices, vectors with partially known support and many others. We will present some of these applications in a few forthcoming papers.

%\newpage1
%\setcounter{page}{1}
\begin{singlespace}
\bibliographystyle{plain}
\bibliography{LiftLqThrBndsRefs}
\end{singlespace}

\end{document}